\documentclass[12pt,a4paper,notitlepage]{article}
\usepackage{amsmath}
\usepackage{amssymb}
\usepackage{amsfonts}
\usepackage{pstricks}
\usepackage[onehalfspacing,nodisplayskipstretch]{setspace}
\usepackage[left=1.25in, right=1.25in, top=1.25in, bottom=1.25in]{geometry}

\newtheorem{theorem}{Theorem}

\newtheorem{corollary}[theorem]{Corollary}

\newtheorem{lemma}[theorem]{Lemma}

\newtheorem{proposition}[theorem]{Proposition}

\newenvironment{proof}[1][Proof]{\noindent\textbf{#1.} }{\hfill $\square$}
\input{tcilatex}
\begin{document}

\title{Forecast-Hedging and Calibration\thanks{%
Previous versions: April 2016; November 2019 (Center for Rationality
DP-731); June 2020. We thank Benjy Weiss for useful discussions, John Levy,
Efe Ok, Sylvain Sorin, and Bernhard von Stengel for references related to
Theorem \ref{th:hairy FP}, and the editor and referees for very helpful
suggestions.}}
\author{Dean P. Foster\thanks{%
Department of Statistics, Wharton, University of Pennsylvania, Philadelphia,
and Amazon, New York. \emph{e-mail}: \texttt{dean@foster.net} \emph{web page}%
: \texttt{http://deanfoster.net}} \and Sergiu Hart\thanks{%
Department of Economics, Department of Mathematics, and the Federmann Center
for the Study of Rationality, The Hebrew University of Jerusalem. \emph{%
e-mail}: \texttt{hart@huji.ac.il} \emph{web page}: \texttt{%
http://www.ma.huji.ac.il/hart}}}
\date{September 21, 2021}
\maketitle

\begin{abstract}
Calibration means that forecasts and average realized frequencies are close.
We develop the concept of forecast hedging, which consists of choosing the
forecasts so as to guarantee that the expected track record can only
improve. This yields all the calibration results by the same simple basic
argument, while differentiating between them by the forecast-hedging tools
used: deterministic and fixed point based versus stochastic and minimax
based. Additional contributions are an improved definition of continuous
calibration, ensuing game dynamics that yield Nash equilibria in the long
run, and a new calibrated forecasting procedure for binary events that is
simpler than all known such procedures.\newpage 
\end{abstract}

\tableofcontents

%TCIMACRO{%
%\TeXButton{References Without Numbers}{\def\@biblabel#1{#1\hfill}
%\def\thebibliography#1{\section*{References}
%\addcontentsline{toc}{section}{References}
%\list
%{}{
%\labelwidth 0pt
%\leftmargin 1.8em
%\itemindent -1.8em
%\usecounter{enumi}}
%\def\newblock{\hskip .11em plus .33em minus .07em}
%\sloppy\clubpenalty4000\widowpenalty4000
%\sfcode`\.=1000\relax\def\baselinestretch{1}\large \normalsize}
%\let\endthebibliography=\endlist}}%
%BeginExpansion
\def\@biblabel#1{#1\hfill}
\def\thebibliography#1{\section*{References}
\addcontentsline{toc}{section}{References}
\list
{}{
\labelwidth 0pt
\leftmargin 1.8em
\itemindent -1.8em
\usecounter{enumi}}
\def\newblock{\hskip .11em plus .33em minus .07em}
\sloppy\clubpenalty4000\widowpenalty4000
\sfcode`\.=1000\relax\def\baselinestretch{1}\large \normalsize}
\let\endthebibliography=\endlist%
%EndExpansion
%TCIMACRO{%
%\TeXButton{define shfrac}{\newcommand{\shfrac}[2]{\ensuremath{{}^{#1} \hspace{-0.04in}/_{\hspace{-0.03in}#2}}}}}%
%BeginExpansion
\newcommand{\shfrac}[2]{\ensuremath{{}^{#1} \hspace{-0.04in}/_{\hspace{-0.03in}#2}}}%
%EndExpansion
%TCIMACRO{%
%\TeXButton{def \T \B}{\newcommand\T{\rule{0pt}{2.6ex}}
%\newcommand\B{\rule[-1.2ex]{0pt}{0pt}}}}%
%BeginExpansion
\newcommand\T{\rule{0pt}{2.6ex}}
\newcommand\B{\rule[-1.2ex]{0pt}{0pt}}%
%EndExpansion

\section{Introduction\label{s:intro}}

Weather forecasters nowadays no longer say that \textquotedblleft it will
rain tomorrow" or \textquotedblleft it will not rain tomorrow"; rather, they
state that \textquotedblleft the chance that it will rain tomorrow is $x.$"
As long as $x$ lies strictly between $0$ and $1,$ they cannot be proven
wrong tomorrow, whether it rains or not. However, they can be proven wrong
over time. This is the case when a forecast, say $x=70\%,$ is repeated many
times, and the proportion of rainy days among those days when the forecast
was $70\%$ is far from $70\%.$

A forecaster is said to be (classically) \emph{calibrated }if, in the long
run, the actual proportions of rainy days are close to the forecasts
(formally, the average difference between frequencies and forecasts---the
calibration score---is small). A surprising result of Foster and Vohra
(1998) shows that one may always generate forecasts that are \emph{%
guaranteed to be calibrated}, no matter what the weather will actually be.%
\footnote{%
There are many proofs of the classic calibration result, some relatively
simple: besides Foster and Vohra (1998), see Hart (1995) (presented in
Section 4 of Foster and Vohra 1998), Foster (1999), Foster and Vohra (1999),
Fudenberg and Levine (1999), Hart and Mas-Colell (2000, 2013), and the
survey of Olszewski (2015).} These forecasts must necessarily be\emph{\
stochastic}; i.e., in each period the forecast $x$ is chosen by a
randomization\footnote{%
That may depend on the history of weather and forecasts.} (e.g., with
probability $1/3$ the forecaster announces that the chance of rain tomorrow
is $x=70\%,$ and with probability $2/3$ the forecaster announces that the
chance is $x=50\%),$ since deterministic forecasts cannot be calibrated
against \emph{all} possible future rain sequences\footnote{\label%
{ftn:rain-iff-<1/2}Consider the sequence where each day there is rain if and
only if the forecast of rain is less than $50\%.$} (cf. Dawid 1982 and Oakes
1985). The analysis is thus from a \textquotedblleft worst-case" point of
view, which is the same as if one were facing an adversarial
\textquotedblleft rain-maker."\footnote{%
Which connects to the related literature on the \textquotedblleft
manipulability of tests"; see Dekel and Feinberg (2006), Olszewski and
Sandroni (2008), and the survey of Olszewski (2015).}

Now the calibration score is discontinuous with respect to the forecasts, as
it considers days when the forecast was, say, $69.9\%,$ separately from the
days when the forecast was $70\%.$ Smoothing out the calibration score by
combining, in a continuous manner, the days when the forecast was \emph{%
close to}\textbf{\ }$x$ before comparing the frequency of rain to $x$ yields
a \emph{continuous calibration }score, which we introduce in Section \ref%
{sus:binning-cc}. The advantage of continuous calibration is that it may be
guaranteed by \emph{deterministic }forecasts (i.e., after every history
there is a single $x$ that is forecasted---in contrast to a probabilistic
distribution over $x$ in the classic calibration setup of the previous
paragraph). Similar concepts that appear in the literature, weak calibration
(Kakade and Foster 2004, Foster and Kakade 2006) and smooth calibration
(Foster and Hart 2018), are encompassed by continuous calibration (see
Appendix \ref{sus-a:smooth-calib}). While the existing proofs of
deterministic smooth and weak calibration are complicated, in the present
paper we provide a simple proof of deterministic continuous
calibration---and so of smooth and weak calibration as well. We thus propose
continuous calibration as the more appropriate concept: more natural, and
easier to analyze and guarantee.

In the present paper we identify specific conditions, which we refer to as 
\emph{forecast-hedging} conditions, that guarantee that the calibration
score will essentially not increase, whatever tomorrow's weather will be.%
\footnote{%
The use of the term \textquotedblleft hedging" here is akin to its use in
finance, where one deals with portfolios that are hedged against risks (by
using, say, appropriate options and derivatives).} Roughly speaking, they
amount to making sure that today's calibration errors will tend to go in the
opposite direction of past calibration errors (thus overshooting, where the
forecast is higher than the frequency of rain, is followed by undershooting,
and the other way around). This is illustrated in Section \ref%
{sus:illustration} below by a stylized simple version of forecast-hedging in
the basic binary rain/no rain setup. Interestingly, it turns out to yield a
new calibrated procedure in this one-dimensional case that is as simple as
can be (and is simpler than the one in Foster 1999); see Section \ref%
{s:1-dim} for the formal analysis.

We show, first, that the main calibration results in the literature
(classic, smooth, weak, almost deterministic, and continuous, introduced
here) all follow from the same simple argument based on forecast-hedging.
Second, we provide the appropriate forecast-hedging tools. In the classic
calibration setup, they correspond to optimal strategies in finite
two-person zero-sum games, whose existence follows from von Neumann's (1928) 
\emph{minimax} theorem, and which are mixed (i.e., stochastic) in general.
In the continuous calibration setup, they correspond to fixed points of
continuous functions, whose existence follows from Brouwer's (1912) \emph{%
fixed point} theorem, and which are deterministic. We refer to the resulting
procedures as procedures of \emph{type MM} and \emph{type FP}, respectively.
This forecast-hedging approach integrates the existing calibration results
by deriving them all from the same proof scheme, while clearly
differentiating\emph{\ }between the MM-procedures and the FP-procedures,
both in terms of the tool they use---minimax vs. fixed point---and in terms
of being stochastic vs. deterministic. Thus classic calibration is obtained
by MM-procedures, whereas continuous calibration, as well as almost
deterministic calibration, by FP-procedures. A further benefit of our
approach is the simple and straightforward proof that it provides of
deterministic continuous calibration, and thus of deterministic smooth
calibration (in contrast to the long and complicated existing proof).

While calibration is stated in terms of \textquotedblleft forecasting," our
forecast-hedging makes it clear that this is a misnomer, as there is no
actual prediction of rain or no rain tomorrow (indeed, such a prediction
cannot be accomplished without making some assumptions on the behavior of
the rain-maker). Rather, calibration obtains by what can be referred to as
\textquotedblleft \emph{back}casting" (instead of \emph{fore}casting):
forecast-hedging guarantees that the past track record can essentially only
improve, no matter what the weather will be.

\subsection{The Economic Utility of Calibration\label{sus:value-calib}}

Now, why would one consider calibration at all? Though some forecasts are
created just for fun (say, predicting a sports winner or a presidential
election), other forecasts drive decision making (say, predicting the chance
of rain or the chance of selling a million widgets). We will focus on
forecasts that have decisions attached to them. If the forecaster is the
same person as the decision maker then he can interpret the forecast in any
fashion he likes and still be consistent. But, when the forecaster is
different from the decision maker, it is desirable for them to be speaking
the same language. To make this concrete, consider the rain forecast that a
traveler hears on landing in a new city. Should an umbrella be unpacked and
made ready? Or is the weather nice enough not to need one? Locals may be
perfectly happy with a forecast that implies some set $\mathcal{U}$ such
that if $x\in \mathcal{U}$ then carrying an umbrella makes sense.\footnote{%
That is, the expected benefit of not being wet on a rainy day exceeds the
expected cost of carrying the umbrella---and perhaps losing it
someplace---on a sunny day.} But, pity our poor traveler who has to figure
out the set $\mathcal{U}$ without any history. Contrast this with the world
where the forecast in each city is known to be calibrated. Then our traveler
can figure out a rule, say, $x>70\%$, and dig his umbrella out if the
forecast is higher than $70\%$. Further, this works for both the timid
traveler who has a rule $x>20\%$ and the outdoors person with a rule of $%
x>99\%$. There can be many other wonderful properties of forecasts that we
could hope to have (accuracy or martingality to name two), but by merely
having calibration the forecasts are connected enough to outcomes to be
useful to decision makers.

Calibration thus allows one to separate the problem into two pieces: the
first is providing a forecast of the world, and the second is taking an
action that is rational given that forecast. This model is a good way of
factoring a business since a forecasting team doesn't need to understand the
nuances that go into the decision making, nor does the decision team need to
know the details of the most current statistical methods that go into making
the forecasts. There are details that the forecasting team will be
continuously worrying about, like whether a neural net is more accurate than
a decision tree or a simple regression. Likewise there are details that the
decision making team will be stressing over, like changing costs and
updating constraints. But, as long as they are communicating via calibrated
forecasts, these worries don't need to be exposed to the other team. The
forecasting team generates calibrated forecasts, and the optimization team
treats these forecasts as if they were probabilities and solves their
optimization problem. This factorization localizes information but still
generates a globally optimal outcome.\footnote{%
A real-life story from a large online retailer is that an old-fashioned ARMA
forecasting model was used for years. It was not calibrated and so the
optimization team had learned to buy more than the forecast suggested. When
the ARMA model was replaced by a modern neural net that was much more
accurate and also calibrated, the retailer lost money---until the
optimization team caught up with the change in the forecasting model. If
both forecasts had been calibrated, there would have been much less internal
stress, and the newer model would have been an easy immediate improvement.}
For a concrete example, consider Figure \ref{fig:bankruptcy}, from Foster
and Stine (2004). 
%TCIMACRO{\TeXButton{BeginFigure}{\begin{figure}[htbp] \centering}}%
%BeginExpansion
\begin{figure}[htbp] \centering%
%EndExpansion
\input{epsf} \epsfxsize=5.5in \epsfbox{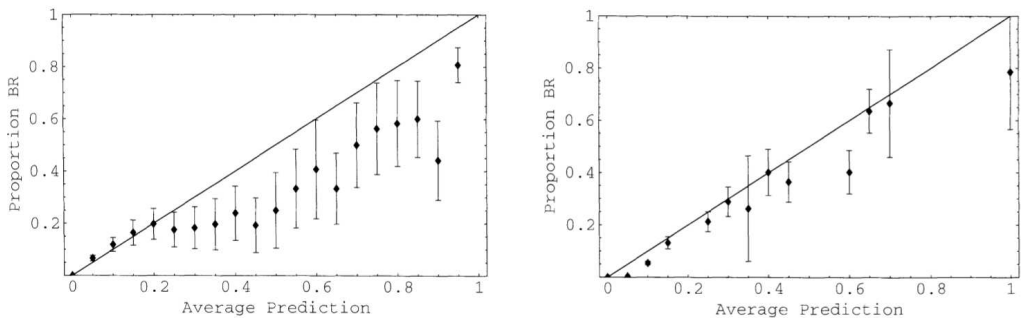}%
\caption{In Foster and Stine
(2004) the business problem was to forecast the chance of a person going
bankrupt in the next month. Both of the above forecasts are based on a large
linear model. The one on the left was obtained by a logistic regression; the one on the right, by a monotone
regression. The left-hand forecast is not calibrated, whereas the right-hand forecast is calibrated and so can be used directly for decision
making.\label{fig:bankruptcy}}%
%TCIMACRO{\TeXButton{EndFigure}{\end{figure}}}%
%BeginExpansion
\end{figure}%
%EndExpansion
It shows two forecasts of when a customer will go bankrupt. The calibrated
forecast (right side) is easy to use: a customer with a forecasted high
chance of bankruptcy shouldn't be extended further credit. The cutoff point
can be created using the costs and benefits to the firm. By contrast,
constructing a rule based on the uncalibrated forecast (left side) requires
actually doing some statistics to figure out what a forecast of, say,
\textquotedblleft $70\%$\textquotedblright\ means. The optimization team
would have to do some empirical statistics, and thus we have failed at
factoring the problem into two clean pieces.

Figure \ref{fig:bankruptcy} may incorrectly suggest that all we need to do
is map a forecast through an appropriate link function that gives the
corresponding average realization and all will be well. This is true for
cross-sectional data and for time-series data where the link function is
evaluated at a single point in time. But, in general, we would need
different such functions at different points in time. Phrased in terms of
our intrepid traveler, if he arrives for a second time at the same foreign
city, the rule he used on the first visit may no longer apply. But, if the
forecasts were calibrated, the same trivial rule would work for both visits.
Mathematically, this means that a calibrated forecast \emph{must} divide an
arbitrary sequence into a collection of subsequences (one for each forecast
value),\footnote{%
We refer to this as \textquotedblleft binning"; see Section \ref%
{sus:binning-cc}.} \emph{all} of which have a limit. This is the hard part.
The fact that we also require a calibrated forecast to know what this limit
is on each of these subsequences is a small restriction compared to
guaranteeing that there are no fluctuations over time and all these limits
exist.

Let us turn to the decision side of the problem. Sometimes the forecast is
so strong for rain,\footnote{%
While we continue to phrase the discussion in terms of rain for simplicity,
think of more meaningful circumstances, such as contextual bandits in
machine learning and personalized medicine in clinical trials.} that not
carrying an umbrella would entail a huge cost. Likewise, it might be that
the chance of rain is so low that carrying one would be too costly. Both of
these costs are relative to the best possible action one could take. But,
sometimes, the forecast is close to the fence and it doesn't really matter
which action is taken. This indifference (equipoise in bio-statistics)
allows one to consider randomizing between these two actions. This would
cheaply allow estimating the actual costs of each action. It would allow one
to compare what would happen if the counterfactual action were taken to what
happens if the action that is believed to be the correct action is taken.
For these reasons, there are many arguments for randomizing at the boundary.
Mathematically it can be thought of as continuously switching from taking an
umbrella (at the boundary plus epsilon) to never taking an umbrella (at the
boundary minus epsilon). If such a continuous response function is used,
then the classic definition of calibration is stronger than it needs to be.
Indeed, we only care about what the approximate value of the forecast is
since we will behave similarly for all such values. This is where continuous
calibration comes in.

Now what is the advantage of using a weaker notion of calibration
(continuous calibration is implied by classic calibration), which is also
more difficult to obtain (it requires a fixed point rather than a minimax
computation every period; see Section \ref{sus:fp-mm procedures}). The
answer is that weakening the calibration requirement allows one to achieve
the important property of \emph{leakiness} of Foster and Hart (2018);
namely, the forecasts remain calibrated even if the action in each period
depends on the forecast (which is the case when the forecast is revealed,
i.e., \textquotedblleft leaked," before the action is chosen). Indeed, for
deterministic procedures that yield continuous calibration, the fact that at
the start of each period $t$ the forecast at $t$ is already known (as it is
fully determined by the history before $t)$ does \emph{not} matter, as
continuous calibration is guaranteed for \emph{any} action. By contrast, for
stochastic procedures that yield classic calibration, at the start of period 
$t$ only the \emph{distribution }of the random forecast at $t$ is known, and 
\emph{not} its actual\emph{\ realization};\emph{\ }if the actual realization
were known, there would be action choices that would invalidate calibration,
as in footnote \ref{ftn:rain-iff-<1/2}. This distinction is underscored by
forecast-hedging, which holds for sure in the deterministic case, and only
in expectation in the stochastic case. It is just as in a two-person
zero-sum game, where an optimal mixed strategy is no longer optimal if the
opponent knows its pure realization, whereas an optimal pure\emph{\ }%
strategy remains so even if known (the same holds for mixed vs. pure Nash
equilibria). So to answer our question, we can trade off this weaker
requirement of calibration for a guarantee of leakiness. Since the weakening
doesn't decrease the value of the forecast for decision making, we have
gained leakiness at minimal cost.

Leakiness turns out to be the crucial property that is needed for \emph{game
dynamics} in general $n$-person games to give Nash equilibria rather than
correlated equilibria. Specifically,\footnote{%
The statements here should be understood with appropriate \textquotedblleft
approximate" adjectives throughout.} while best replying to calibrated
forecasts yields correlated equilibria as the long-run time average of play
(see Foster and Vohra 1997), we show in Section \ref{sus:cont-learn} that
best replying to deterministic continuously calibrated forecasts yields Nash
equilibria being played in most of the periods (see Kakade and Foster 2004
and Foster and Hart 2018 for earlier, somewhat more complicated, variants of
this result).

To return to forecasting, in numerous situations Bayesian methods are
optimal.\footnote{%
Dawid (1982) discusses the connection of calibration to posterior
probabilities, whereas here we want to connect it to the priors.} But, if
you are using the wrong prior, a lot of the charm of Bayesian methods is
lost and estimators that provide robust minimax protection might be
preferred. If we could estimate the prior, then a Bayesian approach sounds
pretty good. This is one of the motivations for empirical Bayesian methods
(see Berger 1985). Unfortunately, unless we are observing a sequence of
independently and identically distributed problems for which we can truly
believe there is a single prior that is common across a string of problems
(see Robbins 1956), then figuring out the prior to use for the next problem
is not easy. This is where calibration can play a part (see George and
Foster 2000). By guaranteeing the connection between the beliefs (our
forecasts) and the actual parameters, we can use a calibrated forecast to
make stronger claims about priors that are estimated in a sequential
empirical Bayes setting.

For a statistician or econometrician, not being calibrated is one of the
most embarrassing mistakes to make. Suppose we are trying to predict some
variable $Y$ based on a bunch of $X_{i}$'s. If it turns out that we could
get a much better fit by looking at $X_{17}/X_{12}$ than we currently are
getting, that would be considered a great scientific result and no one would
fault the previous work that missed it. But, if $3\widehat{Y},$ or $\widehat{%
Y}^{3},$ were better forecasts than the $\widehat{Y}$ provided by the
statistician, that would be an embarrassing error. Given the numerous ways
of correcting uncalibrated forecasts (see Zadrozny and Elkan 2001), people
would ask, \textquotedblleft Didn't you look at your forecast at all?" Of
course, when dealing with out-of-sample forecasts this can occur since the
world might change. Hence, the value of these calibration methods, which
sequentially adapt to a changing world, is to ensure we can avoid this
embarrassment.

Finally, regarding forecast-hedging: as it is an elementary principle, it
might perhaps help dispel some of the mystery behind the prevalence of
well-calibrated forecasts, such as the \textquotedblleft superforecasters"
of the Good Judgement Project (see Tetlock and Gardner 2015 and Mellers et
al. 2015), FiveThirtyEight (see Figure \ref{fig:538}), 
%TCIMACRO{\TeXButton{BeginFigure}{\begin{figure}[htbp] \centering}}%
%BeginExpansion
\begin{figure}[htbp] \centering%
%EndExpansion
\input{epsf} \epsfxsize=5.5in \epsfbox{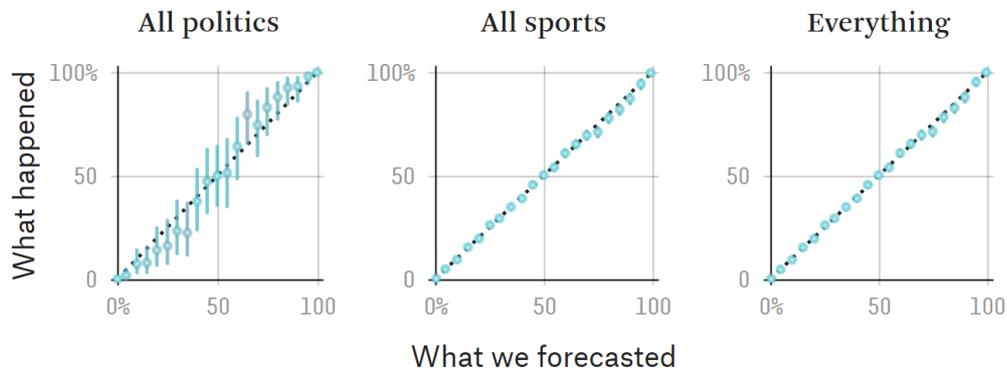}%
\caption{Calibration plots of
FiveThirtyEight (\texttt{projects.fivethirtyeight.com/checking-our-work}, updated on June 26, 2019). For example, 
in the
\textbf{Everything}  plot the $10\%$ data point (which lies slightly below
the diagonal) has the following attached description: \textquotedblleft We
thought the $107962$ observations in this bin had a $10\%$ chance of
happening. They happened $9\%$ of the time."\label{fig:538}}%
%TCIMACRO{\TeXButton{EndFigure}{\end{figure}} }%
%BeginExpansion
\end{figure}
%EndExpansion
ElectionBettingOdds\footnote{%
In such betting / market models, we see that calibration goes part way
toward the \textquotedblleft weak efficient market hypothesis" (wEMH). For
example, take the sequence of times where a stock price is above its
seven-day average and we are considering whether to buy it
(\textquotedblleft momentum") or sell it (\textquotedblleft mean
reversion"). If we had a forecast of the \textquotedblleft correct
price\textquotedblright\ then these could be expressed as saying
\textquotedblleft buy" when the forecast is above the price and
\textquotedblleft sell" when it is below. The property we would then want
such a forecast to have is merely calibration. Given how simple it is for
forecast-hedging to generate calibration, it is reasonable to expect many
traders to all discover something close to the same calibrated forecast and
hence push the market in that direction until the price is the same as the
forecast (while this would not generate the full wEMH, which requires its
holding for all\emph{\ }price patterns, it does go in that direction).} (see
Figure \ref{fig:betting}), 
%TCIMACRO{\TeXButton{BeginFigure}{\begin{figure}[htbp] \centering}}%
%BeginExpansion
\begin{figure}[htbp] \centering%
%EndExpansion
\input{epsf} \epsfxsize=5.5in \epsfbox{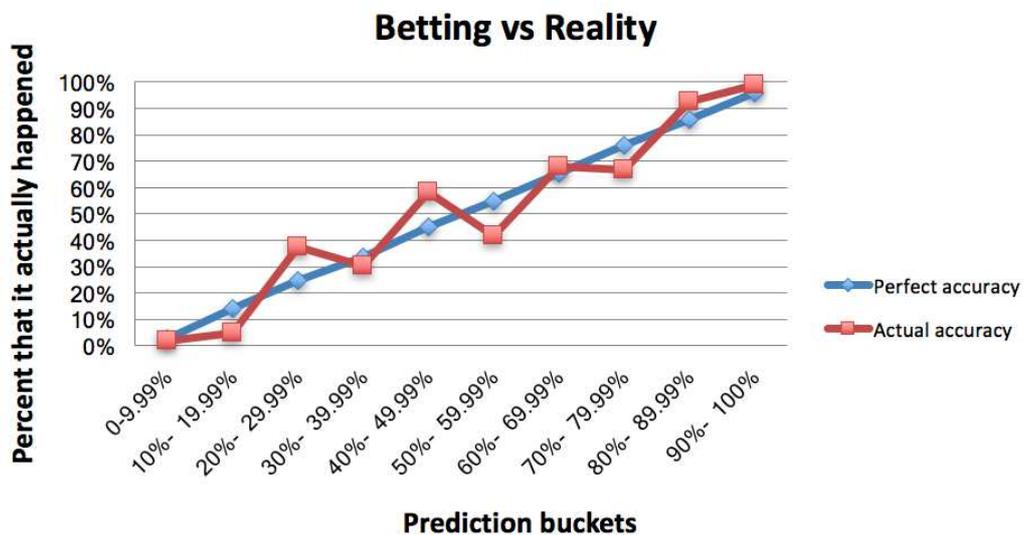}%
\caption{Calibration plot of
ElectionBettingOdds (\texttt{electionbettingodds.com/TrackRecord.html}, updated on November 13, 2018),
which \textquotedblleft  tracked some 462 different candidate chances across dozens of races
 and states in 2016 and 2018."\label{fig:betting}}%
%TCIMACRO{\TeXButton{EndFigure}{\end{figure}} }%
%BeginExpansion
\end{figure}
%EndExpansion
and others. Indeed, in most of these cases one forecasts binary yes/no
events, where, as we show in Sections \ref{sus:illustration} and \ref%
{s:1-dim}, forecast-hedging is extremely simple and straightforward to
implement.\footnote{%
Of course, we are not implying that forecast-hedging is what these
forecasters consciously do. What we are saying is that since calibration is
very easy to achieve, we should not be surprised by its being often
obtained. At the same time, it might be of interest to check if there is any
balancing of current and past forecasting errors, as in forecast-hedging
(see the discussion above where forecast-hedging is defined, and the
illustration in Section \ref{sus:illustration}). Finally, we note that
forecasters are tested not only by their calibration scores, but by stronger
measures of \textquotedblleft accuracy" or \textquotedblleft skill"
(specifically, their Brier scores).}

\subsection{Forecast-Hedging: A Simple Illustration\label{sus:illustration}}

Consider the basic rain/no rain setup---or, for that matter, any sequence of
arbitrary, possibly unrelated, yes/no events (as in the above-mentioned
projects)---and let the forecasts lie on the equally spaced grid $%
0,~1/N,~2/N,~...,~1$ for some integer $N\geq 1.$ Take period $t.$ For each
forecast $x$ let $n(x)\equiv n_{t-1}(x)$ be the number of days that $x$ has
been used in the past $t-1$ periods, and let $r(x)\equiv r_{t-1}(x)$ be the
number of rainy days out of those $n(x)$ days. If the forecast $x$ is
correct there should have been rain on $x\cdot n(x)$ out of the $n(x)$ days,
and so the excess number of rainy days at $x$ is\footnote{%
Think of $G(x)$ as the total \textquotedblleft gap" at $x;$ it may be
positive, zero, or negative. The vertical distance from the diagonal in the
calibration plot (as in Figures 1 and 2) is the normalized gap $G(x)/n(x).$} 
$G(x)\equiv G_{t-1}(x):=r(x)-x\cdot n(x)$. For simplicity consider the sum
of squares score\footnote{\label{ftn:t^2}We abstract away from technical
details, such as the appropriate normalizations, in this illustration; see
Sections \ref{s:calibrated procedures} and \ref{s:1-dim} for the precise
analysis. For the expert reader we note that the calibration score at time $%
t $ is $K_{t}=\sum_{x}|G_{t}(x)|/t$ (see Section \ref{s:calibration-def}),
which is small when $S_{t}/t^{2}$ is small (by the Cauchy--Schwartz
inequality). Note that a constant forecast of, say, $c=1/2$ yields $%
S_{t}=t^{2}/4$ in the worst case (where all days are rainy, or all days are
sunny), and thus a calibration score that is bounded away from zero.} $%
S\equiv S_{t-1}:=\sum_{x}G(x)^{2}.$

Let $a\equiv a_{t}$ denote the weather at time $t,$ with $a=1$ standing for
rain and $a=0$ for no rain, and let $c\equiv c_{t}$ in the interval $[0,1]$
denote the forecast at time $t.$ The change in the score $S$ from time $t-1$
to time $t$ is $S_{t}-S_{t-1}=(G(c)+a-c)^{2}-G(c)^{2}$ (the only term that
changes in the sum $S$ is the $G(c)$ term for the forecasted $c),$ whose
first-order approximation equals $2\Delta $ for\footnote{%
We ignore the term $(a-c)^{2}$, which is bounded by $1,$ since the total
contribution to $S_{t}$ of all these terms is at most $t,$ and thus
negligible relative to $t^{2}$ (see footnote \ref{ftn:t^2}).}%
\begin{equation}
\Delta 
%TCIMACRO{\TeXButton{:=}{{\;:=\;}}}%
%BeginExpansion
{\;:=\;}%
%EndExpansion
G(c)\cdot (a-c).  \label{eq:gradient}
\end{equation}%
We would like to choose the forecast $c$ so that 
\begin{equation}
\Delta \equiv G(c)\cdot (a-c)\leq 0\;\;\;\text{for \emph{any} }a,
\label{eq:det-FH}
\end{equation}%
i.e., no matter what the weather will be. This is easy to do when there is a
point $c$ on the grid with $G(c)=0$: just forecast this $c.$ In general,
however, we can aim only to make the inequality $\Delta \leq 0$ hold \emph{%
on average},\emph{\ }by choosing the forecast at random:\footnote{%
This condition is reminiscent of the Blackwell (1956) approachability
condition in the regret-based approach to calibration of Hart and Mas-Colell
(2000).} 
\begin{equation}
\mathbb{E}\left[ \Delta \right] \equiv \mathbb{E}\left[ G(c)\cdot (a-c)%
\right] \leq 0\;\;\;\text{for \emph{any} }a.  \label{eq:Hedging}
\end{equation}%
This is what we call the \emph{forecast-hedging}\textbf{\ }condition
(condition (\ref{eq:det-FH}) is a special case of (\ref{eq:Hedging})).
Interestingly, this inequality seems to express the idea discussed in the
Introduction that the errors $a-c$ of the current forecast would tend to
have the opposite sign of the errors $G(c)$ of the past forecasts.

How can (\ref{eq:Hedging}) be obtained? Randomizing between two forecasts,
say $c_{1}$ with probability $p_{1}$ and $c_{2}$ with probability $%
p_{2}=1-p_{1},$ yields%
\begin{eqnarray*}
\mathbb{E}\left[ \Delta \right] &=&p_{1}G(c_{1})\cdot
(a-c_{1})+p_{2}G(c_{2})\cdot (a-c_{2}) \\
&=&[p_{1}G(c_{1})+p_{2}G(c_{2})]\cdot (a-c_{2})+p_{1}G(c_{1})\cdot
(c_{2}-c_{1}).
\end{eqnarray*}%
We can thus guarantee $\mathbb{E}\left[ \Delta \right] $ to be small, no
matter what $a$ will be, by choosing the $c_{k}$ and $p_{k}$ so that, first, 
\begin{equation}
p_{1}G(c_{1})+p_{2}G(c_{2})=0,  \label{eq:E[G]=0}
\end{equation}%
and, second, $c_{2}-c_{1}$ is small.\footnote{%
This turns out to suffice because $c_{2}-c_{1}$ is multiplied by $%
p_{1}G(c_{1}),$ which is of the order of $t;$ again, see Sections \ref%
{s:calibrated procedures} and \ref{s:1-dim} for details. The size of the
calibration error is determined by the distance between $c_{1}$\ and $c_{2}$.%
}

Specifically, working on the grid $0,~1/N,~2/N,~...,~1,$ we obtain these
forecasts as follows. If $G(j/N)=0$ for some $j,$ then take $c=j/N,$ which
makes $\Delta =0$. Otherwise $G(i/N)\neq 0$ for all $i,$ and so let $j\geq 1$
be any index with $G(j/N)<0$ (such a $j$ exists because $G(0)>0$ and $%
G(1)<0) $ and take $c_{1}=(j-1)/N$ and $c_{2}=j/N$ (and thus $%
G(c_{1})>0>G(c_{2})),$ with the $p_{k}$ inversely proportional to $%
|G(c_{k})| $ (i.e., as given by (\ref{eq:E[G]=0})). Figure \ref{fig1} 
%TCIMACRO{\TeXButton{BeginFigure}{\begin{figure}[tbp] \centering}}%
%BeginExpansion
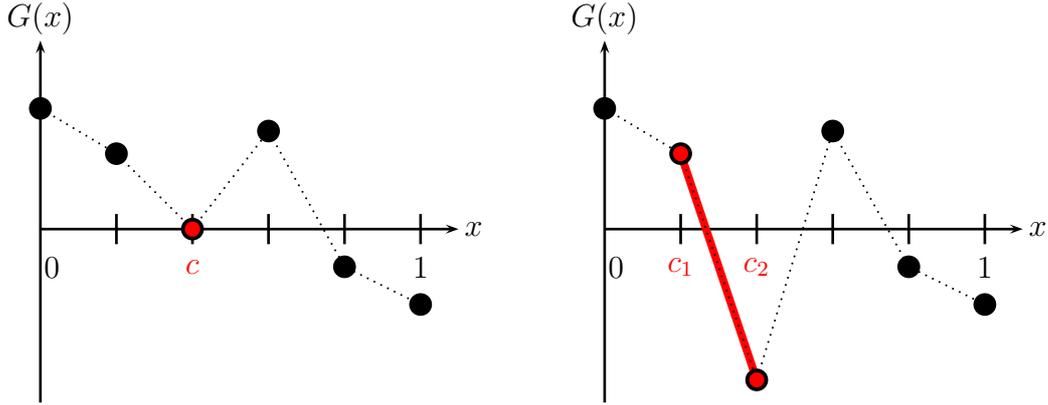
\begin{figure}[tbp] \centering%
%EndExpansion
\begin{pspicture}(-0.5,1)(12,6.5)
    \psset{xunit=0.5cm}
    \psline[linewidth=1pt]{->}(0,4)(11,4)\rput(11.4,4){$x$}
    \psline[linewidth=1pt]{->}(0,1.7)(0,6.5)\rput(0,6.8){$G(x)$}
    \rput(0.3,3.5){$0$}
    \rput(10,3.5){$1$}
    \psline[linewidth=1pt]{}(10,3.8)(10,4.2)
%\FS{3}
    %\rput(0,6.5){$y$}
    \psline[linewidth=1pt]{}(2,3.8)(2,4.2)
    \psline[linewidth=1pt]{}(4,3.8)(4,4.2)
    \psline[linewidth=1pt]{}(6,3.8)(6,4.2)
    \psline[linewidth=1pt]{}(8,3.8)(8,4.2)
%\FS{4}
    \psline[linestyle=dotted,dotsep=2pt](0,5.6)(2,5)(4,4)(6,5.3)(8,3.5)(10,3)
    \pscircle*[linecolor=black](0,5.6){0.15}
    \pscircle*[linecolor=black](2,5){0.15}
    \pscircle*[linecolor=black](4,4){0.15}
    \pscircle*[linecolor=black](6,5.3){0.15}
    \pscircle*[linecolor=black](8,3.5){0.15}
    \pscircle*[linecolor=black](10,3){0.15}
    %\pscircle*[linecolor=black](0,6){0.15}
    %\psline[linewidth=1pt]{}(-0.1,5.6)(0.1,5.6)
    %\onlySlide*{2}{\psline[linewidth=3pt,linecolor=darkred]{-c}(0,1)(3.9,1)
%\FS{5}
    %\psline[linewidth=3pt,linecolor=darkred]{c-}(2,5)(4,2)
    %\pscircle*[linecolor=darkred](2.67,4){0.15}
    %\rput(1.5,7.3){{\darkred $\textbf{1}_x(y)$}}}
%\FS{3}
    %\psline[linewidth=3pt,linecolor=blue]{}(0,1)(2,1)(4,5.6)(6,1)(9.5,1)
    %\psline[linewidth=1pt]{}(2,0.8)(2,1.2)
    %\psline[linewidth=1pt]{}(6,0.8)(6,1.2)
    \rput(4,3.5){{\red $c$}}
    %\rput(4,3.5){$q_2$}
    \pscircle*[linecolor=red](4,4){0.1}
    %\rput(3.9,7.3){{\blue $\Lambda_x(y)=[1-L||x-y||]_+$}}
    \end{pspicture}\hspace*{-2in}%
\begin{pspicture}(-0.5,1)(12,6.5)
    \psset{xunit=0.5cm}
    \psline[linewidth=1pt]{->}(0,4)(11,4)\rput(11.4,4){$x$}
    \rput(0.3,3.5){$0$}
    \rput(10,3.5){$1$}
    \psline[linewidth=1pt]{}(10,3.8)(10,4.2)
%\FS{3}
    %\rput(0,6.5){$y$}
    \psline[linewidth=1pt]{->}(0,1.7)(0,6.5)
    \rput(0,6.8){$G(x)$}
    \psline[linewidth=1pt]{}(2,3.8)(2,4.2)
    \psline[linewidth=1pt]{}(4,3.8)(4,4.2)
    \psline[linewidth=1pt]{}(6,3.8)(6,4.2)
    \psline[linewidth=1pt]{}(8,3.8)(8,4.2)
\psline[linewidth=3pt,linecolor=red]{c-}(2,5)(4,2)
%\FS{4}
    \psline[linestyle=dotted,dotsep=2pt](0,5.6)(2,5)(4,2)(6,5.3)(8,3.5)(10,3)
    \pscircle*[linecolor=black](0,5.6){0.15}
    \pscircle*[linecolor=black](2,5){0.15}
    \pscircle*[linecolor=black](4,2){0.15}
    \pscircle*[linecolor=black](6,5.3){0.15}
    \pscircle*[linecolor=black](8,3.5){0.15}
    \pscircle*[linecolor=black](10,3){0.15}
    %\pscircle*[linecolor=black](0,6){0.15}
    %\psline[linewidth=1pt]{}(-0.1,5.6)(0.1,5.6)
    %\onlySlide*{2}{\psline[linewidth=3pt,linecolor=darkred]{-c}(0,1)(3.9,1)
%\FS{5}
    %\psline[linewidth=3pt,linecolor=red]{c-}(2,5)(4,2)
    %\pscircle*[linecolor=darkred](2.67,4){0.15}
    %\rput(1.5,7.3){{\darkred $\textbf{1}_x(y)$}}}
%\FS{6}
    %\psline[linewidth=3pt,linecolor=blue]{}(0,1)(2,1)(4,5.6)(6,1)(9.5,1)
    %\psline[linewidth=1pt]{}(2,0.8)(2,1.2)
    %\psline[linewidth=1pt]{}(6,0.8)(6,1.2)
    \rput(2,3.5){{\red $c_1$}}
    \rput(4,3.5){{\red $c_2$}}
    %\pscircle*[linecolor=violet](2.67,4){0.1}
    \pscircle*[linecolor=red](4,2){0.1}
    \pscircle*[linecolor=red](2,5){0.1}
    %\rput(3.9,7.3){{\blue $\Lambda_x(y)=[1-L||x-y||]_+$}}
    \end{pspicture}%
\caption{Examples of forecast-hedging. The forecasts $x$
are marked on the horizontal axis, and the gaps $G(x)$ on the vertical axis.
On the left, deterministic forecast-hedging is obtained by forecasting $c$
with $G(c)=0$; on the right, stochastic forecast-hedging is obtained by
randomizing among the forecasts $c_1$ and $c_2$ with probabilities $p_1$ and
$p_2$ such that $p_1G(c_1)+p_2G(c_2)=0$. \label{fig1}}%
%TCIMACRO{\TeXButton{EndFigure}{\end{figure}} }%
%BeginExpansion
\end{figure}
%EndExpansion
provides two examples of graphs of $G$ (for $N=6;$ the dotted lines provide
linear interpolation). On the left we have $c$ on the grid with $G(c)=0,$
which yields the perfect deterministic hedging of (\ref{eq:det-FH}), and on
the right we have adjacent $c_{1},c_{2}$ on the grid with $%
G(c_{1})>0>G(c_{2}),$ which yields the approximate stochastic hedging of (%
\ref{eq:Hedging}).

One of course needs to keep track of all the approximation errors, but,
surprisingly, the procedure described here does work: it guarantees an
average calibration error that goes to $0$ as the grid size $N$ increases;
see Section\footnote{%
We provide there a slight variant that uses the normalized errors $e$
instead of the gaps $G$; it is just as simple, and guarantees the minimal
possible calibration error of $1/(2N)$ (whereas, for the $G$-based procedure
here, the error is of the order of $1/\sqrt{N}).$} \ref{s:1-dim}. It turns
out to be a new addition to the literature, and simpler than any existing
calibrated procedure in this one-dimensional binary setup (i.e., rain/no
rain). Moreover, while it involves randomizations (it must!), the
randomizations are all between two neighboring grid points ($(j-1)/N$ and $%
j/N),$ and so this procedure is an \emph{almost deterministic} procedure
(Foster 1999, Kakade and Foster 2008).

Forecast-hedging is central to all the calibration results in the present
paper, in higher dimensions as well. Specifically:

\begin{itemize}
\item For classic calibration, probabilistic weights $p_{k}$ that ensure
that $\mathbb{E}\left[ \Delta \right] $ is small are obtained using a
minimax result; this is \emph{stochastic} \emph{forecast-hedging}.

\item For continuous calibration, where the corresponding function $G$
becomes continuous, a deterministic point $c$ that ensures $\Delta \leq 0$
(a special case of which is $G(c)=0)$ is obtained by a fixed point result;
this is \emph{deterministic forecast-hedging}.

\item Again for classic calibration, an almost deterministic forecast is
obtained by replacing the fixed point with an appropriate distribution on
nearby grid points (as done above in the one-dimensional case); this is 
\emph{almost deterministic forecast-hedging.}
\end{itemize}

\subsection{The Organization of the Paper}

The paper is organized as follows. Section \ref{s:calibration-def} presents
the general calibration setup, and introduces the new concept of
\textquotedblleft continuous calibration." Section \ref{s:FH tools} is
devoted to what we call \textquotedblleft outgoing" theorems, which provide
the forecast-hedging tools that are used to obtain the calibration results
in Section \ref{s:calibrated procedures}. The simple procedure in the
one-dimensional case is given in Section \ref{s:1-dim}. In Section \ref%
{s:dynamics} we show that the game dynamics of best replying to continuously
calibrated forecasts---\textquotedblleft continuously calibrated
learning"---yield Nash equilibria, and we conclude in Section \ref{sus:fp-mm
universes} with the significant distinction made here between the minimax
and the fixed point universes. The Appendix provides further details,
proofs, and extensions.

\section{The Calibration Setup\label{s:calibration-def}}

Let $A$ be a set of possible outcomes, which we call \emph{actions} (such as 
$A=\{0,1\},$ with $a=1$ standing for rain and $a=0$ for shine), and let $C$
be the set of \emph{forecasts} about these actions (such as $C=[0,1],$ with $%
c$ in $C$ standing for \textquotedblleft the chance of rain is $c$"). We
assume that\footnote{$\mathbb{R}^{m}$ denotes the $m$-dimensional Euclidean
space, with the usual Euclidean ($\ell ^{2}$) norm $\left\Vert \cdot
\right\Vert .$} $C\subset \mathbb{R}^{m}$ is a nonempty compact convex
subset of a Euclidean space, and $A\subseteq C.$ Some special cases of
interest are: (i) $C$ is the set of probability distributions $\Delta (A)$
over a finite set $A,$ which is identified with the set of unit vectors in $%
C $ (and then $C$ is a unit simplex); (ii) $C$ is the convex hull $\mathrm{%
conv}(A)$ of a finite set of points $A\subset \mathbb{R}^{m}$ (and then $C\ $%
is a polytope); and (iii) $C=A$ (and then $A$ is already convex). Let $%
\gamma :=\mathrm{diam}(C)\equiv \max_{c,c^{\prime }\in C}\left\Vert
c-c^{\prime }\right\Vert $ be the diameter of the set $C.$ Let $\delta >0;$
a subset $D$ of $C$ is a $\delta $-\emph{grid} of $C$ if for every $c\in C$
there is $d\in D$ at distance less than $\delta $ from $c,$ i.e., $%
\left\Vert d-c\right\Vert <\delta ;$ a compact set $C$ always has a finite $%
\delta $-grid (obtained from a finite subcover by open $\delta $-balls).

For each period $t=1,2,...,$ let $c_{t}$ in $C$ be the forecast, and let $%
a_{t}$ in $A$ be the action. The forecast at time $t$ may well depend on the
history $h_{t-1}=(c_{1},a_{1};c_{2},a_{2};...;c_{t-1},a_{t-1})\in (C\times
A)^{t-1}$ of past forecasts and actions. A \emph{deterministic }(\emph{%
forecasting}) \emph{procedure }$\sigma $ is thus a mapping $\sigma :\cup
_{t\geq 1}(C\times A)^{t-1}\rightarrow C$ that assigns to every history $%
h_{t-1}$ a forecast $c_{t}=\sigma (h_{t-1})\in C$ at time $t.$ A \emph{%
stochastic }(\emph{forecasting}) \emph{procedure }$\sigma $ is a mapping $%
\sigma :\cup _{t\geq 1}(C\times A)^{t-1}\rightarrow \Delta (C)$ that assigns
to every history $h_{t-1}$ a probability distribution $\sigma (h_{t-1})$ on $%
C$ according to which the forecast $c_{t}$ at time $t$ is chosen. Let $\rho
>0;$ a stochastic procedure $\sigma $ is $\rho $-\emph{almost deterministic}
if for every history $h_{t-1}$ the support of the distribution $\sigma
(h_{t-1})$ of $c_{t}$ is included in a closed ball of radius $\rho $; that
is, the forecast $c_{t}$ is \textquotedblleft deterministic within a
precision $\rho ."$

We refer to $\mathbf{a}=(a_{t})_{t=1}^{\infty },$ where $a_{t}\in A$ for
every $t,$ as an \emph{action sequence}; the sequence may be anything from a
fixed, \textquotedblleft oblivious," sequence, all the way to an adaptive,
\textquotedblleft adversarial," sequence; the latter allows the action $%
a_{t} $ at time $t$ to be determined by the history $h_{t-1},$ as well as by
the forecasting procedure (i.e., the mapping\footnote{\label{ftn:strategies}%
See the remark following the definition of forecast-hedging in Section \ref%
{sus:FH}. In the setup of the \emph{calibration game} (see Foster and Hart
2018), which is a repeated simultaneous game of perfect monitoring and
perfect recall between the action player and the calibrating player (the
forecaster), the statement \textquotedblleft for every action sequence $%
\mathbf{a}$" translates to \textquotedblleft for every (pure) strategy of
the action player."} $\sigma ).$ Let $\mathbf{a}_{t}=(a_{s})_{s=1}^{t}$
denote the first $t$ coordinates of $\mathbf{a}.$

\subsection{Classic Calibration\label{sus:classic calib-def}}

Fix a time $t\geq 1$ and a sequence $(c_{s},a_{s})_{s=1,...,t}\in (C\times
A)^{t}$ of forecasts and actions up to time $t.$ For every $x$ in $C$ let%
\footnote{%
We write $\mathbf{1}_{x}$ for the $x$-indicator function, i.e., $\mathbf{1}%
_{x}(c)=1$ for $c=x$ and $\mathbf{1}_{x}(c)=0$ for $c\neq x.$ The number of
elements of a finite set $Z$ is denoted by $|Z|.$}%
\begin{equation*}
n_{t}(x)%
%TCIMACRO{\TeXButton{:=}{{\;:=\;}}}%
%BeginExpansion
{\;:=\;}%
%EndExpansion
|\{1\leq s\leq t:c_{s}=x\}|=\sum_{s=1}^{t}\mathbf{1}_{x}(c_{s})
\end{equation*}%
be the number of times that the forecast $x$ has been used, and, for every $%
x $ with $n_{t}(x)>0,$ let%
\begin{equation*}
\overline{a}_{t}(x)%
%TCIMACRO{\TeXButton{:=}{{\;:=\;}}}%
%BeginExpansion
{\;:=\;}%
%EndExpansion
\frac{1}{n_{t}(x)}\sum_{s=1}^{t}\mathbf{1}_{x}(c_{s})\,a_{s}
\end{equation*}%
be the average of the actions in all the periods that the forecast $x$ has
been used. The \emph{calibration error} $e_{t}(x)$ of the forecast $x$ is
then%
\begin{equation*}
e_{t}(x)%
%TCIMACRO{\TeXButton{:=}{{\;:=\;}}}%
%BeginExpansion
{\;:=\;}%
%EndExpansion
\overline{a}_{t}(x)-x
\end{equation*}%
(when the forecast $x$ has not been used, i.e., $n_{t}(x)=0,$ we put for
convenience $e_{t}(x)%
%TCIMACRO{\TeXButton{:=}{{\;:=\;}}}%
%BeginExpansion
{\;:=\;}%
%EndExpansion
0).$

The classic \emph{calibration score }is the average calibration error,
namely,\footnote{%
The sum is finite as it goes over all $x$ with $n_{t}(x)>0,$ i.e., over $x$
in the set $\{c_{1},...,c_{t}\}$. In line with standard statistics usage,
one may average the \emph{squared} Euclidean norms $||e_{t}(x)||^{2}$
instead (cf. $X_{t}$ in the proof of Theorem \ref{th:FH}(S)); this will not
affect the results.}%
\begin{equation}
K_{t}%
%TCIMACRO{\TeXButton{:=}{{\;:=\;}}}%
%BeginExpansion
{\;:=\;}%
%EndExpansion
\sum_{x\in C}\left( \frac{n_{t}(x)}{t}\right) \left\Vert e_{t}(x)\right\Vert
;  \label{eq:K}
\end{equation}%
thus, the error of each $x$ is weighted in proportion to the number of times 
$n_{t}(x)$ that $x$ has been used (the weights add up to $1$ because $%
\sum_{x}n_{t}(x)=t$).

Let $\varepsilon >0;$ a (stochastic) procedure $\sigma $ is $\varepsilon $%
\emph{-calibrated} (Foster and Vohra 1998) if\footnote{%
The calibration score $K_{t}$ depends on the actions and forecasts up to
time $t,$ and is thus a function $K_{t}\equiv K_{t}(\mathbf{a},\sigma )$ of
the action sequence $\mathbf{a}$ and the forecasting procedure $\sigma $ (in
fact, only $\mathbf{a}_{t}$ and $\sigma _{t},$ the restriction of $\sigma $
to histories up to time $t,$ matter for $K_{t}$). The same applies to the
other scores throughout the paper.}%
\begin{equation*}
\varlimsup_{t\rightarrow \infty }\left( \sup_{\mathbf{a}_{t}}\mathbb{E}\left[
K_{t}\right] \right) \leq \varepsilon
\end{equation*}%
(the expectation $\mathbb{E}$ is taken over the random forecasts of $\sigma
).$ In Appendix \ref{sus-a:prob1} we show that one may make $K_{t}$ small
with probability one (i.e., almost surely), not just in expectation.

\subsection{Binning and Continuous Calibration\label{sus:binning-cc}}

The calibration error $e_{t}(x)$ can be rewritten as%
\begin{equation*}
e_{t}(x)=\sum_{s=1}^{t}\left( \frac{\mathbf{1}_{x}(c_{s})}{n_{t}(x)}\right)
(a_{s}-c_{s})
\end{equation*}%
(because $\sum_{s=1}^{t}\mathbf{1}_{x}(c_{s})c_{s}=n_{t}(x)x);$ thus, $%
e_{t}(x)$ is the average of the differences $a_{s}-c_{s}$ between actions
and forecasts, where only the periods $s$ where the forecast was $x$ count.

The calibration score, as defined by (\ref{eq:K}), can then be interpreted
as follows. For each $x$ in $C$ there is a \emph{bin}, call it the
\textquotedblleft $x$-bin," which tracks the errors of the forecast $x$;
namely, if at time $s$ the forecast is $c_{s}$ and the action is $a_{s},$
then the difference $a_{s}-c_{s}$ between the action and the forecast is
assigned to the $c_{s}$-bin. At time $t$ one computes the average error $%
e_{t}(x)$ of each $x$-bin, and then the calibration score $K_{t}$ is the
average norm of these errors, where the weight of each $x$-bin is
proportional to its size, i.e., to the number of elements $n_{t}(x)$ that it
contains.

As discussed in the Introduction, the resulting calibration score is highly
discontinuous: forecasts $c$ and $c^{\prime },$ even when slightly apart,
are tracked separately, in distinct bins. To smooth this out and treat them
similarly, we have to, first, allow for \textquotedblleft fractional"
assignments into bins, and, second, make these assignments depend
continuously on the forecast $c$.

What then is a general binning system? It is given by the fraction $0\leq
w_{i}(c)\leq 1$ of each forecast $c$ that goes into each bin $i,$ where
these fractions add up to $1$ over all bins (for each $c).$ We assume for
convenience that the number of bins is countable, i.e., finite or countably
infinite; there is no loss of generality in this assumption, as we show in
Appendix \ref{sus-a:cont-calib}. A \emph{binning }is thus a collection $\Pi
=(w_{i})_{i=1}^{I},$ with $I$ finite or $I=\infty ,$ of functions $%
w_{i}:C\rightarrow \lbrack 0,1]$ such that 
\begin{equation*}
\sum_{i=1}^{I}w_{i}(c)=1
\end{equation*}%
for every $c\in C;$ the binning is \emph{continuous} if all the functions $%
w_{i}$ are continuous functions of $c$.

A continuous binning is obtained, for instance, by taking points $y_{i}$ in $%
C$ and letting the fraction of forecast $c$ that goes into the $y_{i}$-bin
decrease continuously with the distance between $c$ and $y_{i}.$ For a
specific example in which only small neighborhoods matter, take $%
\{y_{1},...,y_{I}\}$ to be a finite $\delta $-grid of $C$ and put $%
w_{i}(c):=\Lambda (c,y_{i})/\sum_{j=1}^{I}\Lambda (c,y_{j})$ for each $1\leq
i\leq I,$ where\footnote{%
For fixed $y,$ the graph of the so-called \textquotedblleft tent function" $%
\Lambda (c,y)$ looks like the symbol $\Lambda $ (with the peak at $c=y$).} $%
\Lambda (c,y):=[\delta -\left\Vert c-y\right\Vert ]_{+}.$

Next, what is the calibration score $K_{t}^{\Pi }$ with respect to a
(continuous) binning $\Pi =(w_{i})_{i=1}^{I}$ ? As for the classic
calibration score $K_{t},$ one first computes the average error in each bin,
and then takes the average norm of these errors, in proportion to the total
weights accumulated in the bins. The total weight of bin $i$ is 
\begin{equation*}
n_{t}^{i}%
%TCIMACRO{\TeXButton{:=}{{\;:=\;}}}%
%BeginExpansion
{\;:=\;}%
%EndExpansion
\sum_{s=1}^{t}w_{i}(c_{s}),
\end{equation*}%
the average error of bin $i$ is 
\begin{equation*}
e_{t}^{i}%
%TCIMACRO{\TeXButton{:=}{{\;:=\;}}}%
%BeginExpansion
{\;:=\;}%
%EndExpansion
\sum_{s=1}^{t}\left( \frac{w_{i}(c_{s})}{n_{t}^{i}}\right) (a_{s}-c_{s})
\end{equation*}%
(again, put $e_{t}^{i}%
%TCIMACRO{\TeXButton{:=}{{\;:=\;}}}%
%BeginExpansion
{\;:=\;}%
%EndExpansion
0$ when $n_{i}^{t}=0),$ and the $\Pi $-calibration score is\footnote{%
For a continuous binning, the bin errors $e_{t}^{i}$ are continuous averages
of the classic calibration errors $e_{t}(x),$ namely,%
\begin{equation*}
e_{t}^{i}=\sum_{x\in C}\left( \frac{w_{i}(x)n_{t}(x)}{n_{t}^{i}}\right)
e_{t}(x);
\end{equation*}%
thus, continuous binnings do indeed capture the idea of smoothing out the
calibration errors (as in the above example with the $\delta $-grid $%
\{y_{i}\}$ on $C).$}%
\begin{equation}
K_{t}^{\Pi }%
%TCIMACRO{\TeXButton{:=}{{\;:=\;}}}%
%BeginExpansion
{\;:=\;}%
%EndExpansion
\sum_{i=1}^{I}\left( \frac{n_{t}^{i}}{t}\right) \left\Vert
e_{t}^{i}\right\Vert  \label{eq:K-Pi}
\end{equation}%
(the weights $n_{t}^{i}/t$ add up to $1,$ because $\sum_{i=1}^{I}n_{t}^{i}=%
\sum_{s=1}^{t}\sum_{i=1}^{I}w_{i}(c_{s})=t).$

A deterministic procedure $\sigma $ is $\Pi $\emph{-calibrated} if%
\begin{equation*}
\lim_{t\rightarrow \infty }\left( \sup_{\mathbf{a}_{t}}K_{t}^{\Pi }\right)
=0,
\end{equation*}%
and it is \emph{continuously calibrated} if it is $\Pi $-calibrated for
every continuous binning $\Pi .$

Compared with classic calibration, continuous calibration requires the
convergence to be to zero (rather than $\leq \varepsilon ),$ simultaneously
for all continuous\footnote{%
One could get uniformity over binnings $\Pi $ by restricting them to a
compact space (for instance, by imposing a uniform Lipschitz condition on
the $w_{i},$ as in weak and smooth calibration).} $\Pi $.

\subsection{Gaps and Preliminary Results\label{sus:gaps}}

Rather than working with the normalized errors, it is convenient to work
with unnormalized \textquotedblleft gaps." For every real function on $C,$
i.e., $w:C\rightarrow \mathbb{R},$ and $t\geq 1,$ let%
\begin{equation*}
g_{t}(w)%
%TCIMACRO{\TeXButton{:=}{{\;:=\;}}}%
%BeginExpansion
{\;:=\;}%
%EndExpansion
\frac{1}{t}\sum_{s=1}^{t}w(c_{s})(a_{s}-c_{s})
\end{equation*}%
be the (\emph{per-period})\emph{\ gap} at time $t$ with respect to $w$ (when 
$w=\mathbf{1}_{x}$ this is the total gap $G(x)$ of Section \ref%
{sus:illustration} divided by the number of periods). We extend the
definitions of $n_{t}$ and $e_{t}$ by 
\begin{equation*}
n_{t}(w)%
%TCIMACRO{\TeXButton{:=}{{\;:=\;}}}%
%BeginExpansion
{\;:=\;}%
%EndExpansion
\sum_{s=1}^{t}w(c_{s})\text{\ \ \ \ and\ \ \ }e_{t}(w)%
%TCIMACRO{\TeXButton{:=}{{\;:=\;}}}%
%BeginExpansion
{\;:=\;}%
%EndExpansion
\frac{1}{n_{t}(w)}\sum_{s=1}^{t}w(c_{s})(a_{s}-c_{s})
\end{equation*}%
for every $w,$ and then the relation $g_{t}(w)=(n_{t}(w)/t)e_{t}(w)$
immediately yields%
\begin{equation*}
K_{t}=\sum_{x\in C}\left\Vert g_{t}(\mathbf{1}_{x})\right\Vert \text{\ \ \ \
and\ \ \ }K_{t}^{\Pi }=\sum_{i=1}^{I}\left\Vert g_{t}(w_{i})\right\Vert
\end{equation*}%
(indeed, for $K_{t}$ we have $n_{t}(x)\equiv n_{t}(\mathbf{1}_{x})$ and $%
e_{t}(x)\equiv e_{t}(\mathbf{1}_{x}),$ and for $K_{t}^{\Pi }$ we have $%
n_{t}^{i}\equiv n_{t}(w_{i})$ and $e_{t}^{i}\equiv e_{t}(w_{i})$).

For every function $w$ the vectors $e_{t}(w)$ and $g_{t}(w)$ are
proportional; they differ in that the denominator is $n_{t}(w)$ in the
former, and $t,$ which is larger, in the latter. The calibration scores are 
\emph{averages} of the norms of $e_{t},$ and \emph{sums} of the norms of $%
g_{t}.$ One advantage of the $g_{t}$ representation is that we do not need
to keep track explicitly of the total weights\footnote{%
In particular, when $n_{t}(w)$ vanishes so does $g_{t}(w).$} $n_{t}.$
Another is that, fixing the sequence of actions and forecasts, the mapping $%
g_{t}$ is a linear bounded operator: $g_{t}(\alpha w+\alpha ^{\prime
}w^{\prime })=\alpha g_{t}(w)+\alpha ^{\prime }g_{t}(w^{\prime })$ for
scalars $\alpha ,\alpha ^{\prime }\in \mathbb{R},$ and, using the supremum
norm $\left\Vert w\right\Vert :=\sup_{c\in C}|w(c)|$ for functions $%
w:C\rightarrow \mathbb{R}$, we have%
\begin{equation*}
\left\Vert g_{t}(w)\right\Vert \leq \gamma \left\Vert w\right\Vert
\end{equation*}%
(because $g_{t}(w)$ is an average of vectors $w(c_{s})(a_{s}-c_{s})$ of norm 
$\leq \left\Vert w\right\Vert \mathrm{diam}(C)=\left\Vert w\right\Vert
\gamma );$ therefore%
\begin{equation}
|~\left\Vert g_{t}(w)\right\Vert -\left\Vert g_{t}(w^{\prime })\right\Vert
~|\leq \left\Vert g_{t}(w)-g_{t}(w^{\prime })\right\Vert =\left\Vert
g_{t}(w-w^{\prime })\right\Vert \leq \gamma \left\Vert w-w^{\prime
}\right\Vert .  \label{eq:norm-g}
\end{equation}

Returning to the binning condition, which we can write as\footnote{%
We write $\mathbf{1}$ for the constant $1$ function; all indicator and $w$
functions are defined on $C$ only.} $\sum_{i=1}^{I}w_{i}=\mathbf{1},$ it
says that $\Pi =(w_{i})_{i=1}^{I}$ is a \textquotedblleft partition of
unity," and so the resulting calibration score $K_{t}^{\Pi }$ may be viewed
as the \textquotedblleft variation" of $g_{t}$ with respect to the partition 
$\Pi $. In particular, the classic calibration score $K_{t}$ is the
variation of $g_{t}$ with respect to the partition $\sum_{x\in C}\mathbf{1}%
_{x}=\mathbf{1}$ into indicator functions. Since the indicator partition is
the finest partition,\footnote{%
Any further split into fractions of indicators does not matter since $%
g_{t}(\alpha \mathbf{1}_{x})=\alpha g_{t}(\mathbf{1}_{x})$.} it stands to
reason that $K_{t}$ would be the maximal possible variation, i.e., the
\textquotedblleft total variation" of $g_{t}.$ This is indeed so: for every
binning $\Pi $ we have%
\begin{equation}
K_{t}^{\Pi }\leq K_{t},  \label{eq:K-Pi<=K}
\end{equation}%
which immediately follows from applying Lemma \ref{l:norm-W} below to $\Pi .$

Thus, any notion based on binning---in particular, continuous
calibration---is a weakening of classic calibration: if $K_{t}$ is small,
then so are all the relevant $K_{t}^{\Pi }.$

\begin{lemma}
\label{l:norm-W}Let $(w_{j})_{j\in J}$ be a countable collection of
nonnegative functions on $C,$ i.e., $w_{j}:C\rightarrow \mathbb{R}_{+}$ for
every $j\in J.$ Then%
\begin{equation*}
\sum_{j\in J}\left\Vert g_{t}(w_{j})\right\Vert \leq \left\Vert \sum_{j\in
J}w_{j}\right\Vert K_{t}.
\end{equation*}
\end{lemma}

\begin{proof}
Put $W:=\sum_{j\in J}w_{j};$ using $w_{j}=\sum_{x\in C}w_{j}(x)\mathbf{1}%
_{x} $ and the linearity of $g_{t}$ we have\footnote{%
The sum $\sum_{x\in C}$ in the proof below is a finite sum (over $x\in
\{c_{1},...,c_{t}\}),$ and so it commutes with $\sum_{j\in J}.$}%
\begin{eqnarray*}
\sum_{j\in J}\left\Vert g_{t}(w_{j})\right\Vert &\leq &\sum_{j\in
J}\sum_{x\in C}w_{j}(x)\left\Vert g_{t}(\mathbf{1}_{x})\right\Vert
=\sum_{x\in C}\sum_{j\in J}w_{j}(x)\left\Vert g_{t}(\mathbf{1}%
_{x})\right\Vert \\
&=&\sum_{x\in C}W(x)\left\Vert g_{t}(\mathbf{1}_{x})\right\Vert \leq
\left\Vert W\right\Vert \sum_{x\in C}\left\Vert g_{t}(\mathbf{1}%
_{x})\right\Vert =\left\Vert W\right\Vert K_{t}.
\end{eqnarray*}
\end{proof}

For another use of this lemma, let $\Pi =(w_{i})_{i=1}^{\infty }$ be an
infinite continuous binning. The increasing sequence of continuous functions 
$\sum_{i=1}^{k}w_{i}$ converges pointwise, as $k\rightarrow \infty ,$ to the
continuous function $\mathbf{1},$ on the compact set $C,$ and so, by Dini's
theorem (see, e.g., Rudin 1976, Theorem 7.13), the convergence is uniform:%
\begin{equation}
\lim_{k\rightarrow \infty }\left\Vert \sum_{i=k+1}^{\infty }w_{i}\right\Vert
=0.  \label{eq:dini}
\end{equation}%
Using Lemma \ref{l:norm-W} for every $\mathbf{a}_{t},$ together with $%
K_{t}\leq \gamma $ (by (\ref{eq:K})), yields%
\begin{equation}
\lim_{k\rightarrow \infty }\left( \sup_{\mathbf{a}_{t}}\sum_{i=k+1}^{\infty
}\left\Vert g_{t}(w_{i})\right\Vert \right) =0.  \label{eq:tail-gt}
\end{equation}%
Thus, for continuous binning only finitely many $w_{i}$ matter, which leads
to a simpler characterization of continuous calibration in terms of
\textquotedblleft pointwise-in-$w"$ convergence.

\begin{proposition}
\label{p:cc<=>g}A deterministic forecasting procedure $\sigma $ is
continuously calibrated if and only if%
\begin{equation}
\lim_{t\rightarrow \infty }\left( \sup_{\mathbf{a}_{t}}\left\Vert
g_{t}(w)\right\Vert \right) =0  \label{eq:G/t-pointwise}
\end{equation}%
for every continuous function $w:C\rightarrow \lbrack 0,1].$
\end{proposition}

\begin{proof}
Given a continuous function $w:C\rightarrow \lbrack 0,1],$ let $\Pi $ be the
continuous binning $(w,\mathbf{1}-w)$. Since $\left\Vert g_{t}(w)\right\Vert
\leq K_{t}^{\Pi },$ continuous calibration implies (\ref{eq:G/t-pointwise})$%
. $

Conversely, let $\Pi =(w_{i})_{i=1}^{I}$ be a continuous binning. When $I$
is finite we have $\sup_{\mathbf{a}_{t}}K_{t}^{\Pi }=\sup_{\mathbf{a}%
_{t}}\sum_{i=1}^{I}\left\Vert g_{t}(w_{i})\right\Vert \leq
\sum_{i=1}^{I}\sup_{\mathbf{a}_{t}}\left\Vert g_{t}(w_{i})\right\Vert ,$
which converges to $0$ as $t\rightarrow \infty $ by (\ref{eq:G/t-pointwise}%
). When $I$ is infinite, for every $\varepsilon >0$ there is by (\ref%
{eq:tail-gt}) a finite $k$ such that $\sup_{\mathbf{a}_{t}}K_{t}^{\Pi }\leq
\sum_{i=1}^{k}\sup_{\mathbf{a}_{t}}\left\Vert g_{t}(w_{i})\right\Vert
+\varepsilon ,$ which converges to $\varepsilon $ as $t\rightarrow \infty $
by (\ref{eq:G/t-pointwise}); since $\varepsilon $ is arbitrary, the limit is 
$0.$
\end{proof}

\bigskip

We now construct a continuous binning $\Pi _{0}$ such that $\Pi _{0}$%
-calibration implies $\Pi $-calibration for \emph{all} continuous $\Pi $
(and so $\Pi _{0}$ plays, for continuous calibration, the same role that the
indicator binning plays for classic calibration; see (\ref{eq:K-Pi<=K})).

\begin{proposition}
\label{p:universal partition}There exists a continuous binning $\Pi _{0}$
such that a deterministic forecasting procedure $\sigma $ is continuously
calibrated if and only if it is $\Pi _{0}$-calibrated$.$
\end{proposition}

\begin{proof}
The space of continuous functions from the compact set $C$ to $[0,1]$ is
separable with respect to the supremum norm; let $(u_{i})_{i=1}^{\infty }$
be a dense sequence. Take $\alpha _{i}>0$ such that $\sum_{i=1}^{\infty
}\alpha _{i}\left\Vert u_{i}\right\Vert \leq 1$ (for example, $\alpha
_{i}=1/(2^{i}\left\Vert u_{i}\right\Vert )$), and put $w_{i}:=\alpha
_{i}u_{i}$ for all $i\geq 1$ and $w_{0}:=\mathbf{1-}\sum_{i=1}^{\infty
}w_{i} $ (the function $w_{0}$ is continuous because $\sum_{i=1}^{\infty
}w_{i}(c)\leq \sum_{i=1}^{\infty }\alpha _{i}\left\Vert u_{i}\right\Vert
\leq 1).$ Thus $\Pi _{0}=(w_{i})_{i=0}^{\infty }$ is a continuous binning,
and so continuous calibration implies $\Pi _{0}$-calibration.

Conversely, $\Pi _{0}$-calibration implies (\ref{eq:G/t-pointwise}) for each 
$w_{i}$ in $\Pi _{0}$ (because $\left\Vert g_{t}(w_{i})\right\Vert \leq
K_{t}^{\Pi _{0}}),$ and hence for each $u_{i}$ by the linearity of $g_{t}.$
This extends from the dense sequence $(u_{i})_{i}$ to any continuous $%
w:C\rightarrow \lbrack 0,1]$ by (\ref{eq:norm-g}), and Proposition \ref%
{p:cc<=>g} completes the proof.
\end{proof}

\bigskip

Proposition \ref{p:cc<=>g} also implies that continuous calibration is a
strengthening of existing Lipschitz-based notions of weak calibration
(Kakade and Foster 2004, Foster and Kakade 2006) and smooth calibration
(Foster and Hart 2018). Indeed, a continuously calibrated procedure---a
simple construction of which we provide in Section \ref{s:calibrated
procedures}---is \textquotedblleft universally" weakly and smoothly
calibrated (by contrast, the known constructions depend on the Lipschitz
bound $L$ and the desired calibration error\footnote{\label{ftn:doubling}The
traditional way to obtain universal procedures is by restarting them at
appropriate times with new values of the parameters (as in Section 4.4 of
Kakade and Foster 2004). The procedures that we construct in the present
paper are much simpler.} $\varepsilon $). See Proposition \ref%
{p:cc=>smooth,weak} in Appendix \ref{sus-a:smooth-calib}. Thus, continuous
calibration may well be used instead of weak and smooth calibration.

\section{Forecast-Hedging Tools\label{s:FH tools}}

In this section we provide useful variants of Brouwer's (1912) fixed point
theorem and von Neumann's (1928) minimax theorem; they are used in Section %
\ref{s:calibrated procedures} to obtain forecasts that satisfy the
forecast-hedging conditions. These conditions, of the form (\ref{eq:det-FH})
and (\ref{eq:Hedging}) (see Section \ref{sus:illustration}), are referred to
as \textquotedblleft outgoing" because of their geometric interpretation
(see the paragraph following the statement of Theorem \ref{th:hairy FP}
below). The reader may skip the proofs in this section at first reading;
however, see the important distinction between fixed point and minimax
procedures in Section \ref{sus:fp-mm procedures}.

Throughout this section $f:C\rightarrow \mathbb{R}^{m}$ is a function from
the nonempty compact and convex subset $C$ of $\mathbb{R}^{m}$ into $\mathbb{%
R}^{m}$ (with the same dimension $m),$ which may be interpreted as a vector
field \textquotedblleft flow" (i.e., think of $x$ as moving to $x+f(x),$ or
to $x+\varepsilon f(x)$ for some $\varepsilon >0).$

\subsection{Outgoing Fixed Point\label{sus:outgoing fp}}

When the function $f$ is \emph{continuous} we get:

\begin{theorem}[Outgoing Fixed Point]
\label{th:hairy FP}Let $C\subset \mathbb{R}^{m}$ be a nonempty compact
convex set, and $f:C\rightarrow \mathbb{R}^{m}$ be a continuous function.
Then there exists a point $y$ in $C$ such that 
\begin{equation}
f(y)\cdot (x-y)\leq 0  \label{eq:OFP}
\end{equation}%
for all $x\in C.$
\end{theorem}

Thus, $f(y)\cdot y=\max_{x\in C}f(y)\cdot x.$ If $y$ is an interior point of 
$C$ then we must have $f(y)=0$ (because $x-y$ can be proportional to any
vector in $\mathbb{R}^{m}),$ and if $y$ is on the boundary of $C$ then $f(y)$
is an outgoing normal to the boundary of\ $C$ at $y$. This result is the
\textquotedblleft variational inequalities" Lemma 8.1 in Border (1985), who
attributes it to Hartman and Stampacchia (1966, Lemma 3.1). We provide a
short direct proof using Brouwer's (1912) fixed point theorem.\footnote{%
Theorem \ref{th:hairy FP} is in fact equivalent to Brower's fixed point
theorem, as the latter is easily proved from the former; see Appendix \ref%
{sus-a:outgoing}, which contains various comments on the \textquotedblleft
outgoing" results.}

\bigskip

\begin{proof}
For every $z\in \mathbb{R}^{m}$ let $\xi (z)\in C$ be the closest point to $%
z $ in the set $C,$ i.e., $\left\Vert \xi (z)-z\right\Vert =\min_{x\in
C}\left\Vert x-z\right\Vert $. As is well known, because $C$ is a convex and
compact set, $\xi (z)$ is well defined (i.e., it exists and is unique), the
function $\xi $ is continuous, and 
\begin{equation}
(z-\xi (z))\cdot (x-\xi (z))\leq 0  \label{eq:normal}
\end{equation}%
for every $x\in C$ (when $z\in C$ it trivially holds because then $\xi
(z)=z, $ and when $z\notin C$ the vector $z-\xi (z)$ is an outward normal to 
$C$ at the boundary point $\xi (z)).$

The function $x\longmapsto \xi (x+f(x))$ is thus a continuous function from $%
C$ to $C$, and so by Brouwer's fixed point theorem there is $y\in C$ such
that $y=\xi (y+f(y)).$ Applying (\ref{eq:normal}) to the point $z=y+f(y),$
for which $\xi (z)=y,$ yields the result.
\end{proof}

\subsection{Outgoing Minimax\label{sus:outgoing mm}}

For functions $f$ that need \emph{not} be continuous we have:

\begin{theorem}[Outgoing Minimax]
\label{th:hairy mM}Let $C\subset \mathbb{R}^{m}$ be a nonempty compact
convex set, let $D\subset C$ be a finite $\delta $-grid of $C$ for some $%
\delta >0,$ and let $f:D\rightarrow \mathbb{R}^{m}.$ Then there exists a
probability distribution $\eta $ on $D$ such that%
\begin{equation}
\mathbb{E}_{y\sim \eta }\left[ f(y)\cdot (x-y)\right] \leq \delta \,\mathbb{E%
}_{y\sim \eta }\left[ \left\Vert f(y)\right\Vert \right]  \label{eq:expfy}
\end{equation}%
for all $x\in C.$ Moreover, the support of $\eta $ can be taken to consist
of at most $m+3$ points of $D.$
\end{theorem}

When $f$ is \emph{bounded}, by taking $\delta =\varepsilon /\sup_{x\in
C}\left\Vert f(x)\right\Vert $ we get:

\begin{corollary}
\label{c:outgoing-MM}Let $C\subset \mathbb{R}^{m}$ be a nonempty compact
convex set, and $f:C\rightarrow \mathbb{R}^{m}$ a bounded function. Then for
every $\varepsilon >0$ there exists a probability distribution $\eta $ on $C$
such that%
\begin{equation*}
\mathbb{E}_{y\sim \eta }\left[ f(y)\cdot (x-y)\right] \leq \varepsilon
\end{equation*}%
for all $x\in C.$ Moreover, the support of $\eta $ can be taken to consist
of at most $m+3$ points of $C.$
\end{corollary}

Unlike in the Outgoing Fixed Point Theorem \ref{th:hairy FP}, in the
Outgoing Minimax Theorem \ref{th:hairy mM} $y$ is a random variable and no
longer a constant, and the \textquotedblleft outgoing" inequality holds in
expectation (within an arbitrarily small error). The proof is a finite
construct that uses the von Neumann's (1928) minimax theorem,\footnote{%
As we will see in Appendix \ref{sus-a:outgoing}, Corollary \ref%
{c:outgoing-MM} is equivalent to the minimax theorem (as Theorem \ref%
{th:hairy FP} is equivalent to Brouwer's fixed point theorem).} and thus
amounts to solving a linear programming problem.

\bigskip

\begin{proof}[Proof of Theorem \protect\ref{th:hairy mM}]
Let $\delta _{0}\equiv \delta _{0}(D):=\max_{x\in C}\mathrm{dist}(x,D)$ be
the farthest away a point in $C$ may be from the $\delta $-grid $D;$ the
maximum is attained on the compact set $C$ and so $\delta _{0}<\delta .$ Put 
$\delta _{1}:=\delta -\delta _{0}>0,$ and take $B\subset C$ to be a finite $%
\delta _{1}$-grid of $C.$ Consider the finite two-person zero-sum game where
the maximizer chooses $b\in B,$ the minimizer chooses $y\in D,$ and the
payoff is\ $f(y)\cdot (b-y)-\delta _{0}\left\Vert f(y)\right\Vert .$ For
every mixed strategy $\nu \in \Delta (B)$ of the maximizer, let $\bar{b}:=%
\mathbb{E}_{b\sim \nu }\left[ b\right] \in C$ be its expectation; the
minimizer can make the payoff $\leq 0$ by choosing a point $y$ on the grid $%
D $ that is within $\delta _{0}$ of $\bar{b}$: 
\begin{equation*}
\mathbb{E}_{x\sim \nu }\left[ f(y)\cdot (b-y)-\delta _{0}\left\Vert
f(y)\right\Vert \right] =f(y)\cdot (\bar{b}-y)-\delta _{0}\left\Vert
f(y)\right\Vert \leq 0
\end{equation*}%
(because $f(y)\cdot (\bar{b}-y)\leq \left\Vert f(y)\right\Vert \cdot
\left\Vert \bar{b}-y\right\Vert \leq \left\Vert f(y)\right\Vert \delta
_{0}). $ Therefore, by the minimax theorem, the minimizer can \emph{guarantee%
} that the payoff is $\leq 0$; i.e., there is a mixed strategy $\eta \in
\Delta (D)$ such that 
\begin{equation}
\mathbb{E}_{y\sim \eta }\left[ f(y)\cdot (b-y)-\delta _{0}\left\Vert
f(y)\right\Vert \right] \leq 0  \label{eq:delta-M}
\end{equation}%
for every $b\in B.$ Since for every $x\in C$ there is $b\in B$ with $%
\left\Vert x-b\right\Vert <\delta _{1},$ and so $f(y)\cdot (x-b)\leq \delta
_{1}\left\Vert f(y)\right\Vert $ for every $y,$ adding this inequality to (%
\ref{eq:delta-M}) yields, by $\delta _{0}+\delta _{1}=\delta ,$ the
inequality (\ref{eq:expfy}) for every $x\in C.$

For the moreover statement, (\ref{eq:expfy}) says that the vector $\mathbb{E}%
_{y\sim \eta }\left[ F(y)\right] $ satisfies $\mathbb{E}_{y\sim \eta }\left[
F(y)\right] \cdot (x,-1,-\delta )\leq 0$ for every $x\in C,$ where 
\begin{equation*}
F(y):=(f(y),f(y)\cdot y,\left\Vert f(y)\right\Vert )\in \mathbb{R}^{m+2}
\end{equation*}%
for each $y\in D.$ By Carath\'{e}odory's theorem, $\mathbb{E}_{y\sim \eta }%
\left[ F(y)\right] $ can be expressed as a convex combination of at most $%
m+3 $ points in $\{F(y):y\in D\},$ and so the support of $\eta $ can be
taken to be of size at most $m+3.$
\end{proof}

\subsection{Almost Deterministic Outgoing Fixed Point\label{sus:outgoing
almost det}}

We can improve the result of the Outgoing Minimax Theorem and obtain a
probability distribution that is \textquotedblleft almost
deterministic"---i.e., the randomization is between nearby points---by using
a fixed point.

A probability distribution $\eta $ is said to be $\rho $-\emph{local} if its
support is included in a closed ball of radius $\rho $; i.e., there exists $%
x $ such that $\eta (\overline{B}(x;\rho ))=1,$ where $B(x;\rho
)=\{z:\left\Vert z-x\right\Vert <\rho \}$ and $\overline{B}(x;\rho
)=\{z:\left\Vert z-x\right\Vert \leq \rho \}$ denote, respectively, the open
and closed balls of radius $\rho $ around $x.$

\begin{theorem}[Almost Deterministic Outgoing Fixed Point]
\label{th:dhfp}Let $C\subset \mathbb{R}^{m}$ be a nonempty compact convex
set, let $D\subset C$ be a finite $\delta $-grid of $C$ for some $\delta >0,$
and let $f:D\rightarrow \mathbb{R}^{m}.$ Then there exists a $\delta $-local
probability distribution $\eta $ on $D$ such that%
\begin{equation*}
\mathbb{E}_{y\sim \eta }\left[ f(y)\cdot (x-y)\right] \leq \delta \,\mathbb{E%
}_{y\sim \eta }\left[ \left\Vert f(y)\right\Vert \right]
\end{equation*}%
for all $x\in C.$ Moreover, the support of $\eta $ can be taken to consist
of at most $m+1$ points of $D.$
\end{theorem}

When $f$ is bounded, by taking $\delta =\min \{\varepsilon /\sup_{x\in
C}\left\Vert f(x)\right\Vert ,\rho \}$ we get:

\begin{corollary}
\label{c:outgoing-ad}Let $C\subset \mathbb{R}^{m}$ be a nonempty compact
convex set, and $f:C\rightarrow \mathbb{R}^{m}$ a bounded function. Then for
every $\varepsilon >0$ and $\rho >0$ there exists a $\rho $-local
probability distribution $\eta $ on $C$ such that%
\begin{equation*}
\mathbb{E}_{y\sim \eta }\left[ f(y)\cdot (x-y)\right] \leq \varepsilon
\end{equation*}%
for all $x\in C.$ Moreover, the support of $\eta $ can be taken to consist
of at most $m+1$ points of $C.$
\end{corollary}

\begin{proof}[Proof of Theorem \protect\ref{th:dhfp}]
From the values of $f$ on $D$ one can generate a continuous function $%
\widetilde{f}:C\rightarrow \mathbb{R}^{m}$ such that $\widetilde{f}(x)$ is a
weighted average of the values of $f$ on grid points that are within $\delta 
$ of $x,$ i.e., 
\begin{equation}
\widetilde{f}(x)\in \mathrm{conv}\{f(d):d\in D\cap B(x;\delta )\},
\label{eq:f-hat}
\end{equation}%
for all $x\in C.$ For instance, put%
\begin{equation*}
\widetilde{f}(x)%
%TCIMACRO{\TeXButton{:=}{{\;:=\;}}}%
%BeginExpansion
{\;:=\;}%
%EndExpansion
\frac{\sum_{d\in D}\Lambda (x,d)f(d)}{\sum_{d\in D}\Lambda (x,d)},
\end{equation*}%
where $\Lambda (x,d):=[\delta -\left\Vert d-x\right\Vert ]_{+}$ (the
so-called $``$tent"\ function); $\widetilde{f}$ is continuous because $D$ is
finite, $\Lambda (x,d)$ is continuous in $x$, and the denominator is always
positive since $D$ is a $\delta $-grid of $C$; as for (\ref{eq:f-hat}), it
follows since $\left\Vert d-x\right\Vert \geq \delta $ implies $\Lambda
(x,d)=0.$

Theorem \ref{th:hairy FP} applied to $\widetilde{f}(x)$ yields a point $z\in
C$ such that $\widetilde{f}(z)\cdot (x-z)\leq 0$ for all $x\in C,$ and then (%
\ref{eq:f-hat}) yields a probability distribution $\eta $ on $D\cap
B(z;\delta )$ such that $\widetilde{f}(z)=\mathbb{E}_{y\sim \eta }[f(y)].$
The distribution $\eta $ is thus $\delta $-local$,$ and its support can be
taken to be of size at most $m+1$ by Carath\'{e}odory's theorem (because $%
f(y)\in \mathbb{R}^{m}).$ Now%
\begin{equation*}
\mathbb{E}_{y\sim \eta }[f(y)\cdot (x-y)]=\mathbb{E}_{y\sim \eta }[f(y)\cdot
(x-z)]+\mathbb{E}_{y\sim \eta }[f(y)\cdot (z-y)];
\end{equation*}%
the first term is $\mathbb{E}_{y\sim \eta }[f(y)]\cdot (x-z)=\widetilde{f}%
(z)\cdot (x-z)\leq 0$ (by the choice of $z),$ and the second term is $\leq
\delta \,\mathbb{E}_{y\sim \eta }\left[ \left\Vert f(y)\right\Vert \right] $
(because $\left\Vert y-z\right\Vert \leq \delta $ for every $y$ in the
support of $\eta ),$ which completes the proof.
\end{proof}

\subsection{FP-Procedures and MM-Procedures\label{sus:fp-mm procedures}}

The calibration proofs that we provide below construct procedures where the
forecast in each period is given by appealing either to the Outgoing Fixed
Point Theorems \ref{th:hairy FP} and \ref{th:dhfp} or to the Outgoing
Minimax Theorem \ref{th:hairy mM}, in order to satisfy the corresponding
forecast-hedging conditions. We will refer to these two kinds of procedures
as \emph{procedures of type FP }and \emph{procedures of type MM},
respectively.

This distinction is not just a matter of proof technique. It goes the other
way around as well (see Hazan and Kakade 2012 for details and relevant
literature): calibration that is obtained by FP-procedures, such as
continuous calibration, may be used to get approximate Nash equilibria in
non-zero-sum games.\footnote{%
This should come as no surprise since game dynamics where players best reply
to continuously calibrated forecasts yield in the long run approximate Nash
equilibria for general $n$-person games; see Section \ref{s:dynamics}.}
Therefore this kind of calibration falls essentially in the PPAD complexity
class, which is believed to go beyond the class of polynomially solvable
problems, such as minimax problems. The distinction between FP-obtainable
calibration and MM-obtainable calibration is a significant distinction, of
the non-polynomial vs. polynomial variety; see also Section \ref{sus:fp-mm
universes}.

\section{Calibrated Procedures\label{s:calibrated procedures}}

In this section we prove the three main calibration results: deterministic
continuous calibration, stochastic classic calibration, and almost
deterministic classic calibration. The proofs all run along the same lines:
first, we show that appropriate forecast-hedging conditions yield
calibration (Theorem \ref{th:FH}); and second, we construct, using the
outgoing results of Section \ref{s:FH tools}, procedures that satisfy the
forecast-hedging conditions (Theorem \ref{th:exist}).

We illustrate the idea of the proof (see also Section \ref{sus:illustration}%
) by showing how to construct a deterministic procedure that guarantees that 
$g_{t}(w)\rightarrow 0$ as $t\rightarrow \infty $ (see (\ref%
{eq:G/t-pointwise})) for a single continuous function $w:C\rightarrow
\lbrack 0,1].$ By the definition of $g_{t}$ we have $%
tg_{t}(w)=(t-1)g_{t-1}(w)+w(c_{t})(a_{t}-c_{t}),$ and so%
\begin{equation}
\left\Vert tg_{t}(w)\right\Vert ^{2}=\left\Vert (t-1)g_{t-1}(w)\right\Vert
^{2}+2(t-1)g_{t-1}(w)\cdot w(c_{t})(a_{t}-c_{t})+w(c_{t})^{2}\left\Vert
a_{t}-c_{t}\right\Vert ^{2}.  \label{eq:tg2}
\end{equation}%
The last term is $\leq \gamma ^{2}$ (since $w(c_{t})\in \lbrack 0,1]$ and $%
a_{t},c_{t}$ belong to $C,$ whose diameter is $\gamma ).$ The middle term is 
$2(t-1)\varphi (c_{t})\cdot (a_{t}-c_{t}),$ where $\varphi
(c):=w(c)g_{t-1}(w)$ is a continuous function of $c$ that takes values in $%
\mathbb{R}^{m}$ (because $g_{t-1}(w)\in \mathbb{R}^{m}).$ The Outgoing Fixed
Point Theorem \ref{th:hairy FP} then yields a point in\footnote{%
In this simple case of a single $w$ a fixed point is not really needed: take 
$c_{t}$ in $C$ that is maximal in the direction $g_{t-1}(w)$, i.e., $%
c_{t}\in \arg \max_{x\in C}x\cdot g_{t-1}(w).$ The fixed point \emph{is}
however needed once we consider multiple $w$'s.} $C$---which will be our
forecast $c_{t}$---that guarantees that $\varphi (c_{t})\cdot
(a_{t}-c_{t})\leq 0,$ for any action $a_{t}\in A\subseteq C.$ Therefore (\ref%
{eq:tg2}) yields the inequality $\left\Vert tg_{t}(w)\right\Vert ^{2}\leq
\left\Vert (t-1)g_{t-1}(w)\right\Vert ^{2}+\gamma ^{2},$ which applied
recursively gives $\left\Vert tg_{t}(w)\right\Vert ^{2}\leq t\gamma ^{2},$
and thus $\left\Vert g_{t}(w)\right\Vert \leq \gamma /\sqrt{t}\rightarrow 0$
as $t\rightarrow \infty .$ The proof is easily extended to handle continuous
binnings $(w_{i})_{i}$, such as $\Pi _{0}$ of Proposition \ref{p:universal
partition}, which yields continuous calibration. For classic calibration,
where the function $\varphi $ above is in general not continuous, we use the
Outgoing Minimax Theorem \ref{th:hairy mM} (for a variant $\psi $ of $%
\varphi );$ finally, using the Outgoing Almost Deterministic Fixed Point
Theorem \ref{th:dhfp} instead yields an almost deterministic procedure for
classic calibration.

\subsection{Forecast-Hedging\label{sus:FH}}

Let $\Pi =(w_{i})_{i=1}^{I}$ be a binning. For every period $t\geq 2$ and
history $h_{t-1}$ we define two functions, $\varphi _{t-1}$ and $\psi
_{t-1}, $ from $C$ to $\mathbb{R}^{m},$ by%
\begin{eqnarray*}
&&\varphi _{t-1}(c)%
%TCIMACRO{\TeXButton{:=}{{\;:=\;}}}%
%BeginExpansion
{\;:=\;}%
%EndExpansion
\sum_{i=1}^{I}w_{i}(c)g_{t-1}(w_{i}) \\
&&\psi _{t-1}(c)%
%TCIMACRO{\TeXButton{:=}{{\;:=\;}}}%
%BeginExpansion
{\;:=\;}%
%EndExpansion
\sum_{i=1}^{I}w_{i}(c)e_{t-1}(w_{i})
\end{eqnarray*}%
for every $c\in C.$ Thus $\varphi _{t-1}$ and $\psi _{t-1}$ are averages of
the vectors $g_{t-1}(w_{i})$ and $e_{t-1}(w_{i}),$ respectively, with
weights that vary with $c$ and are given by the binning $\Pi .$ We define:

\begin{description}
\item \textbf{(D) }A deterministic forecasting procedure $\sigma $ satisfies
the\emph{\ }$\Pi $\emph{-deterministic forecast-hedging} condition if, for
every $t\geq 2$ and history $h_{t-1},$%
\begin{equation}
\varphi _{t-1}(c_{t})\cdot (a-c_{t})\leq 0\ \ \text{for every }a\in A, 
\tag{D-FH}  \label{eq:DFH}
\end{equation}%
where $c_{t}=\sigma (h_{t-1})$ is the forecast at time $t$.

\item \textbf{(S) }A stochastic procedure $\sigma $ satisfies the $(\Pi
,\varepsilon )$-\emph{stochastic forecast}-\emph{hedging} condition for $%
\varepsilon >0$ if, for every $t\geq 2$ and history $h_{t-1},$%
\begin{equation}
\mathbb{E}_{t-1}\left[ \psi _{t-1}(c_{t})\cdot (a-c_{t})\right] \leq
\varepsilon \,\mathbb{E}_{t-1}\left[ \left\Vert \psi
_{t-1}(c_{t})\right\Vert \right] \text{\ \ for every }a\in A,  \tag{S-FH}
\label{eq:SFH}
\end{equation}%
where $\mathbb{E}_{t-1}$ denotes expectation with respect to the
distribution $\sigma (h_{t-1})$ of the forecast $c_{t}$ at time $t$.
\end{description}

\bigskip

\noindent \textbf{Remark. }The forecast-hedging conditions (\ref{eq:DFH})
and (\ref{eq:SFH}) require, for each history $h_{t-1},$ that the
corresponding inequality hold \textquotedblleft for \emph{every} $a\in A."$
This allows the action $a_{t}$ that follows the history $h_{t-1}$ to depend
on $h_{t-1},$ and thus also on $\sigma (h_{t-1}),$ which is determined by $%
h_{t-1}$. Therefore, when $\sigma $ is a deterministic procedure, $a_{t}$
may depend on $c_{t}$ as well; this is the \textquotedblleft leaky" setup of
Foster and Hart (2018) (when $\sigma $ is stochastic it may depend on the
distribution $\sigma (c_{t})$ of $c_{t},$ but \emph{not} on the actual
realization of $c_{t}$). See footnote \ref{ftn:strategies} and Section \ref%
{s:dynamics}.

\begin{theorem}
\label{th:FH} ~

\emph{\textbf{(D)}} If a deterministic procedure $\sigma $ satisfies the $%
\Pi $-deterministic forecast-hedging condition for a \emph{continuous}
binning $\Pi =(w_{i})_{i=1}^{I},$ then%
\begin{equation}
\lim_{t\rightarrow \infty }\left( \sup_{\mathbf{a}_{t}}K_{t}^{\Pi }\right)
=0.  \label{eq:K-Pi->0}
\end{equation}

\emph{\textbf{(S)} If }a stochastic procedure $\sigma $ satisfies the $(\Pi
,\varepsilon )$-stochastic forecast-hedging condition for a \emph{finite}
binning $\Pi =(w_{i})_{i=1}^{I}$ and $\varepsilon >0,$ then%
\begin{equation}
\varlimsup_{t\rightarrow \infty }\left( \sup_{\mathbf{a}_{t}}\mathbb{E}\left[
K_{t}^{\Pi }\right] \right) \leq \varepsilon .  \label{eq:sum-V}
\end{equation}
\end{theorem}

\begin{proof}
\textbf{(D)} Put\footnote{%
The score $S_{t}$ is precisely $S$ of Section \ref{sus:illustration}.} $%
S_{t}:=\sum_{i=1}^{I}\left\Vert tg_{t}(w_{i})\right\Vert ^{2};$ we will show
that $\lim_{t\rightarrow \infty }(1/t^{2})S_{t}=0.$

Using (\ref{eq:tg2}) for each $w_{i},$ summing over $i,$ and recalling the
definition of $\varphi _{t-1}$ gives%
\begin{equation*}
S_{t}\leq S_{t-1}+2(t-1)\varphi _{t-1}(c_{t})\cdot (a_{t}-c_{t})+\gamma ^{2}
\end{equation*}%
(the last term is $\sum_{i}w_{i}(c_{t})^{2}\left\Vert a_{t}-c_{t}\right\Vert
^{2}\leq \gamma ^{2}\sum_{i}w_{i}(c_{t})=\gamma ^{2}$ since $w_{i}(c_{t})\in
\lbrack 0,1]).$ This inequality becomes $S_{t}\leq S_{t-1}+\gamma ^{2}$ when 
$\sigma $ satisfies (\ref{eq:DFH}); by recursion (starting with $S_{0}=0)$
we get $S_{t}\leq t\gamma ^{2}.$ All the inequalities hold for every action
sequence $\mathbf{a}_{t}\mathbf{,}$ because for every history $h_{t-1},$
inequality (\ref{eq:DFH}) holds for every $a.$ Thus, dividing by $t^{2},$ we
have%
\begin{equation*}
\sup_{\mathbf{a}_{t}}\sum_{i=1}^{I}\left\Vert g_{t}(w_{i})\right\Vert
^{2}\leq \frac{\gamma ^{2}}{t}\underset{t\rightarrow \infty }{%
\longrightarrow }0.
\end{equation*}%
Therefore $\sup_{\mathbf{a}_{t}}\left\Vert g_{t}(w_{i})\right\Vert
\rightarrow 0$ as $t\rightarrow \infty $ for every $i\in I,$ which yields (%
\ref{eq:K-Pi->0}) (by the same argument as in the second part of the proof
of Proposition \ref{p:cc<=>g}, because the binning $\Pi $ is continuous).

\textbf{(S)} Put $X_{t}:=\sum_{i=1}^{I}n_{t}(w_{i})\left\Vert
e_{t}(w_{i})\right\Vert ^{2}.$ We will show that 
\begin{equation*}
\varlimsup_{t\rightarrow \infty }\left( \sup_{\mathbf{a}_{t}}\mathbb{E}\left[
\frac{1}{t}X_{t}\right] \right) \leq \varepsilon ^{2};
\end{equation*}%
this yields (\ref{eq:sum-V}) since $\left( K_{t}^{\Pi }\right) ^{2}\leq
(1/t)X_{t}$ by Jensen's inequality.\footnote{%
The score $(1/t)X_{t}$ is the \emph{square}-calibration score for $\Pi ,$
namely, the average of the \emph{squared} norms of the errors (i.e., replace 
$\left\Vert e_{t}^{i}\right\Vert $ with $\left\Vert e_{t}^{i}\right\Vert
^{2} $ in formula (\ref{eq:K-Pi}) of $K_{t}^{\Pi })$.}

The proof consists of expressing the one-period increment of $X_{t}$ as a
sum of two terms, a $Y_{t}$-term, which, by forecast-hedging, is at most $%
\varepsilon ^{2}$ in expectation, and a $Z_{t}$-term, which converges to
zero:%
\begin{eqnarray}
X_{t}-X_{t-1} &=&Y_{t}+Z_{t},  \label{eq:X-X} \\
\mathbb{E}_{t-1}\left[ Y_{t}\right] &\leq &\varepsilon ^{2},\text{\ \ and}
\label{eq:Y-V} \\
\sup_{\mathbf{a}_{t}}\sum_{s=1}^{t}Z_{s} &\leq &\mathrm{O}(\log t)
\label{eq:Z}
\end{eqnarray}%
for every $t\geq 1$ (where $X_{0}=0).$ This proves the result, since taking
overall expectation of (\ref{eq:Y-V}) yields $\mathbb{E}\left[ Y_{t}\right]
\leq \varepsilon ^{2},$ and thus%
\begin{eqnarray*}
\mathbb{E}\left[ \frac{1}{t}X_{t}\right] &=&\mathbb{E}\left[ \frac{1}{t}%
\sum_{s=1}^{t}(X_{s}-X_{s-1})\right] =\frac{1}{t}\sum_{s=1}^{t}\mathbb{E}%
\left[ Y_{s}\right] +\frac{1}{t}\sum_{s=1}^{t}\mathbb{E}\left[ Z_{s}\right]
\\
&\leq &\varepsilon ^{2}+\mathrm{O}\left( \frac{\log t}{t}\right) \rightarrow
\varepsilon ^{2}
\end{eqnarray*}%
as $t\rightarrow \infty ,$ uniformly over $\mathbf{a}_{t}.$

$\bullet $ \emph{Proof of (\ref{eq:X-X}).} We start with the following
easy-to-check identity, for scalars $\alpha ,\beta \geq 0$ and vectors $u,v$:%
\begin{equation*}
(\alpha +\beta )\left\Vert \frac{\alpha u+\beta v}{\alpha +\beta }%
\right\Vert ^{2}-\alpha \left\Vert u\right\Vert ^{2}=2\beta u\cdot v-\beta
\left\Vert u\right\Vert ^{2}+\frac{\beta ^{2}}{\alpha +\beta }\left\Vert
u-v\right\Vert ^{2}.
\end{equation*}%
Using this for $\alpha =n_{t-1}(w),$ $\beta =w(c_{t}),$ $u=e_{t-1}(w),$ and $%
v=a_{t}-c_{t}$ yields%
\begin{equation*}
n_{t}(w)\left\Vert e_{t}(w)\right\Vert ^{2}-n_{t-1}(w)\left\Vert
e_{t-1}(w)\right\Vert ^{2}=y_{t}(w)+z_{t}(w),
\end{equation*}%
where%
\begin{eqnarray*}
y_{t}(w) &%
%TCIMACRO{\TeXButton{:=}{{\;:=\;}}}%
%BeginExpansion
{\;:=\;}%
%EndExpansion
&2w(c_{t})e_{t-1}(w)\cdot (a_{t}-c_{t})-w(c_{t})\left\Vert
e_{t-1}(w)\right\Vert ^{2}\text{\ \ and} \\
z_{t}(w) &%
%TCIMACRO{\TeXButton{:=}{{\;:=\;}}}%
%BeginExpansion
{\;:=\;}%
%EndExpansion
&\frac{w(c_{t})^{2}}{n_{t}(w)}\left\Vert e_{t-1}(w)-(a_{t}-c_{t})\right\Vert
^{2}\leq 4\gamma ^{2}\frac{w(c_{t})^{2}}{n_{t}(w)}
\end{eqnarray*}%
(the last inequality because $\left\Vert e_{t-1}(w)\right\Vert \leq \gamma $
and $\left\Vert a_{t}-c_{t}\right\Vert \leq \gamma ).$ Applying this to each 
$w_{i},$ summing over $i,$ and recalling the definition of $X_{t}$ and $\psi
_{t-1}$ gives $X_{t}-X_{t-1}=Y_{t}+Z_{t},$ where%
\begin{eqnarray*}
&Y_{t}&%
%TCIMACRO{\TeXButton{:=}{{\;:=\;}}}%
%BeginExpansion
{\;:=\;}%
%EndExpansion
\sum_{i=1}^{I}y_{t}(w_{i})=2\psi _{t-1}(c_{t})\cdot
(a_{t}-c_{t})-\sum_{i=1}^{I}w_{i}(c_{t})\left\Vert e_{t-1}(w_{i})\right\Vert
^{2},\;\;\text{and} \\
&Z_{t}&%
%TCIMACRO{\TeXButton{:=}{{\;:=\;}}}%
%BeginExpansion
{\;:=\;}%
%EndExpansion
\sum_{i=1}^{I}z_{t}(w_{i})\leq 4\gamma ^{2}\sum_{i=1}^{I}\frac{%
w_{i}(c_{t})^{2}}{n_{t}(w_{i})}.
\end{eqnarray*}

$\bullet $ \emph{Proof of (\ref{eq:Y-V}).} By the stochastic
forecast-hedging condition (\ref{eq:SFH}) we have $\mathbb{E}_{t-1}\left[
2\psi _{t-1}(c_{t})\cdot (a_{t}-c_{t})\right] \leq \mathbb{E}_{t-1}\left[
2\varepsilon \left\Vert \psi _{t-1}(c_{t})\right\Vert \right] ;$ now 
\begin{eqnarray*}
2\varepsilon \left\Vert \psi _{t-1}(c_{t})\right\Vert &\leq
&\sum_{i=1}^{I}w_{i}(c_{t})\left( 2\varepsilon \left\Vert
e_{t-1}(w_{i})\right\Vert \right) \mathbb{\leq }\sum_{i=1}^{I}w_{i}(c_{t})%
\left( \varepsilon ^{2}+\left\Vert e_{t-1}(w_{i})\right\Vert ^{2}\right) \\
&=&\varepsilon ^{2}+\sum_{i=1}^{I}w_{i}(c_{t})\left\Vert
e_{t-1}(w_{i})\right\Vert ^{2}.
\end{eqnarray*}

$\bullet $ \emph{Proof of (\ref{eq:Z}). }We claim that\footnote{%
One can easily obtain a bound of $\mathrm{o}(t)$ in (\ref{eq:Z->0}), since $%
w(c_{t})^{2}/n_{t}(w)\leq w(c_{t})/n_{t}\rightarrow 0$ as $t\rightarrow
\infty $ (indeed, if $n_{t}(w)\rightarrow \infty $ then $w(c_{t})/n_{t}(w)%
\leq 1/n_{t}(w)\rightarrow 0,$ and if $n_{t}(w)\rightarrow N<\infty $ then $%
w(c_{t})/n_{t}(w)=1-n_{t-1}(w)/n_{t}(w)\rightarrow 1-N/N=0).$ Inequality (%
\ref{eq:Z->0}) provides a better bound, uniform over all $w$ and sequences $%
c_{t}$.} 
\begin{equation}
\sum_{s=1}^{t}\frac{w(c_{s})^{2}}{n_{s}(w)}<\ln n_{t}(w)+2\leq \ln t+2
\label{eq:Z->0}
\end{equation}%
for every $w:C\rightarrow \lbrack 0,1]$ and $t\geq 1$ with $n_{t}(w)>0.$
Indeed, both $w(c_{s})$ and $w(c_{s})/n_{s}(w)$ are between $0$ and $1,$ and
so for every $1\leq r\leq t$ we have%
\begin{eqnarray*}
\sum_{s=1}^{r}\frac{w(c_{s})^{2}}{n_{s}(w)} &\leq
&\sum_{s=1}^{r}w(c_{s})=n_{r}(w)\;\;\;\;\text{and} \\
\sum_{s=r+1}^{t}\frac{w(c_{s})^{2}}{n_{s}(w)} &\leq &\sum_{s=r+1}^{t}\frac{%
w(c_{s})}{n_{s}(w)}=\sum_{s=r+1}^{t}\left( 1-\frac{n_{s-1}(w)}{n_{s}(w)}%
\right)  \\
&\leq &\sum_{s=r+1}^{t}\ln \left( \frac{n_{s}(w)}{n_{s-1}(w)}\right) =\ln
\left( \frac{n_{t}(w)}{n_{r}(w)}\right) 
\end{eqnarray*}%
(we used $1-1/x\leq \ln x$ for $x\geq 1).$ Taking $r\leq t$ such that $1\leq
n_{r}(w)<2$ yields $<2$ in the first inequality and $\leq \ln n_{t}(x)\leq
\ln t$ in the second, and thus (\ref{eq:Z->0}); if there is no such $r$ then 
$n_{t}(w)<1,$ and the first inequality with $r=t$ gives $<1,$ and thus (\ref%
{eq:Z->0}). Applying (\ref{eq:Z->0}) to each $w_{i}$ and summing over $i$
yields $\sum_{s=1}^{t}Z_{s}\leq 4\gamma ^{2}I(\ln t+2)$, and thus (\ref{eq:Z}%
).

This completes the proof of (S).
\end{proof}

\bigskip

\noindent \textbf{Remark.} In (S), using (\ref{eq:Y-V}) one gets the
stronger almost sure convergence; see Appendix \ref{sus-a:prob1}.

\bigskip

The reason that the two proofs are slightly different---we use $S_{t},$ and
thus $\varphi _{t-1},$ in (D), and $X_{t},$ and thus $\psi _{t-1},$ in
(S)---has to do with the limit being $0$ in the former, and $\varepsilon $
in the latter. Roughly speaking, for vectors $u$ in $I$-dimensional space, $%
\left\Vert u\right\Vert =\left( \sum_{i}u_{i}^{2}\right) ^{1/2}\rightarrow 0$
implies $\left\Vert u\right\Vert _{1}=\sum_{i}|u_{i}|\rightarrow 0$
regardless of the size of $I,$ whereas $\left\Vert u\right\Vert \leq
\varepsilon $ yields $\left\Vert u\right\Vert _{1}\leq \sqrt{I}\varepsilon ,$
which may not be small when $I$ increases with $\varepsilon ;$ see Appendix %
\ref{sus-a:FH} for further details.

\bigskip

We now show that the outgoing results of Section \ref{s:FH tools} yield the
existence of forecast-hedging procedures.

\begin{theorem}
\label{th:exist}~

\emph{\textbf{(D)}} For every continuous binning $\Pi $ there exists a
deterministic procedure of type FP that satisfies the $\Pi $-deterministic
forecast-hedging condition.

\emph{\textbf{(S)} }For every finite binning $\Pi $, every $\varepsilon >0,$
and every finite $\varepsilon $-grid $D$ of $C,$ there exists a stochastic
procedure of type MM with forecasts in $D$ that satisfies the $(\Pi
,\varepsilon )$-stochastic forecast-hedging condition.

\emph{\textbf{(AD)} }For every finite binning $\Pi $, every $\varepsilon >0,$
and every finite $\varepsilon $-grid $D$ of $C,$ there exists an $%
\varepsilon $-almost deterministic procedure of type FP with forecasts in $D$
that satisfies the $(\Pi ,\varepsilon )$-stochastic forecast-hedging
condition.
\end{theorem}

\begin{proof}
\textbf{(D)} When $\Pi $ is a continuous binning, each function $\varphi
_{t-1}$ is continuous (since each $w_{i}$ is continuous and $\left\Vert
g_{t-1}(w_{i})\right\Vert \leq \gamma ;$ when $I$ is infinite use the
uniform convergence of the corresponding finite sums, as in the second part
of the proof of Proposition \ref{p:cc<=>g}). Apply the Outgoing Fixed Point
Theorem \ref{th:hairy FP} to $\varphi _{t-1}$ for each history $h_{t-1}.$

\textbf{(S)} Apply the Outgoing Minimax Theorem \ref{th:hairy mM} to $\psi
_{t-1}$ and $\delta =\varepsilon $ for each history $h_{t-1}.$

\textbf{(AD)} Apply the Outgoing Almost Deterministic Fixed Point Theorem %
\ref{th:dhfp} to $\psi _{t-1}$ and $\delta =\varepsilon $ for each history $%
h_{t-1}.$
\end{proof}

\subsection{Calibration}

We now immediately obtain the existence of appropriate calibrated procedures.

\begin{theorem}
\label{th:all-calib}~

\emph{\textbf{(D)}} There exists a deterministic procedure of type FP that
is continuously calibrated.

\emph{\textbf{(S)} }For every $\varepsilon >0$ there exists a stochastic
procedure of type MM that is $\varepsilon $-calibrated; moreover, all its
forecasts are in $D$ for any given finite $\varepsilon $-grid $D$ of $C.$

\emph{\textbf{(AD)} }For every $\varepsilon >0$ there exists an $\varepsilon 
$-almost deterministic procedure of type FP that is $\varepsilon $%
-calibrated; moreover, all its forecasts are in $D$ for any given finite $%
\varepsilon $-grid $D$ of $C.$
\end{theorem}

Part (D) implies, by Proposition \ref{p:cc=>smooth,weak} in Appendix \ref%
{sus-a:smooth-calib}, the results of Foster and Hart (2018) for smooth
calibration and of Kakade and Foster (2004) and Foster and Kakade (2006) for
weak calibration. Part (S) yields the classic calibration result of Foster
and Vohra (1998), and part (AD) the result of Kakade and Foster (2004) for
almost deterministic classic calibration.

\bigskip

\begin{proof}
\textbf{(D)} Apply Theorem \ref{th:exist}(D) and Theorem \ref{th:FH}(D) with
the continuous binning $\Pi _{0}$ given by Proposition \ref{p:universal
partition}.

\textbf{(S)} Let $D=\{d_{1},...,d_{I}\}$ be a given finite $\varepsilon $%
-grid of $C.$ Put $w_{i}:=\mathbf{1}_{d_{i}}$ for $i=1,...,I,$ and $w_{0}:=%
\mathbf{1}_{C\backslash D}$ and let $\Pi $ be the finite binning $%
(w_{i})_{i=0}^{I}.$ When all forecasts are in $D$ we have $K_{t}^{\Pi
}=\sum_{i=1}^{I}\left\Vert g_{t}(\mathbf{1}_{d_{i}})\right\Vert =K_{t}$
(since $g_{t}(w_{0})=0).$ Apply Theorem \ref{th:exist}(S) and Theorem \ref%
{th:FH}(S).

\textbf{(AD)} Same as (S), applying Theorem \ref{th:exist} (AD).
\end{proof}

\section{A Simple Calibrated Procedure for Binary Events\label{s:1-dim}}

This section shows how to obtain classic calibration in the one-dimensional
case, where the actions are binary yes/no outcomes (such as win/lose in
politics and sport events, or rain/shine, and so on), by a procedure that is
as simple as can be; it is simpler than any existing procedure, including
the one in Foster (1999). The procedure is moreover almost deterministic,
with all randomizations being between two neighboring points on a fixed
grid. It is essentially the procedure described in Section \ref%
{sus:illustration} in the Introduction, except that we work with the
normalized errors $e$ instead of the gaps $G.$

We are thus in the one-dimensional case ($m=1$), with $A=\{0,1\}$ (with,
say, $1$ for \textquotedblleft rain" and $0$ for \textquotedblleft no rain")
and $C=[0,1].$ Fix an integer $N\geq 1,$ and let $D:=\{0,~1/N,~2/N,~...,~1\}$
be the grid on which the forecasts lie. Consider a history $h_{t-1}$. For
every $i=0,1,...,N,$ the error of the forecast $i/N$ is $%
e^{i}:=e_{t-1}(i/N)=r^{i}/n^{i}-i/N,$ where $n^{i}$ is the number of times
that the forecast $i/N$ has been used in the first $t-1$ periods, and $r^{i}$
is the number of rainy periods among these $n^{i}$ periods (with $e^{i}=0$
when $n^{i}=0).$ The procedure $\sigma $ chooses the forecast $c_{t}$ as
follows (as in Figure \ref{fig1}, with $e$ instead of $G$):

\begin{itemize}
\item \emph{Case 1: There is }$j$\emph{\ such that }$e^{j}=0.$ Put $y:=j/N$
and let the (deterministic) forecast be\footnote{%
Since $e^{j}=0$ for unused forecasts $j/N,$ in the first periods we try each
point on the grid once; alternatively, assume that there is some initial
data for each possible forecast (all this does not matter, of course, in the
long run).} $c_{t}=y.$

\item \emph{Case 2: }$e^{i}\neq 0$\emph{\ for all }$i.$ In this case $%
e^{0}>0 $ (because $r^{0}\geq 0)$ and $e^{N}<0$ (because $r^{N}\leq n^{N}),$
and so let $j\geq 1$ be, for concreteness, the smallest index with $e^{j}<0;$
thus\footnote{%
Any $j$ for which $e^{j-1}$ and $e^{j}$ have opposite signs will work here.
In fact, a $j$ for which the signs are reversed, i.e., $e_{{}}^{j-1}<0<e^{j}$
(however, such a $j$ need not exist in general), will work even better, as
it yields $0$ on the right-hand side of the forecast-hedging condition (\ref%
{eq:SFH}).} $e^{j-1}>0>e^{j}.$ Put $y_{1}:=(j-1)/N$ and $y_{2}:=j/N,$ and
let the forecast be $c_{t}=y_{1}$ with probability $%
p_{1}:=|e^{j}|/(|e^{j-1}|+|e^{j}|)$ and $c_{t}=y_{2}$ with the remaining
probability $p_{2}:=|e^{j-1}|/(|e^{j-1}|+|e^{j}|);$ thus, $%
p_{1}e_{t-1}(y_{1})+p_{2}e_{t-1}(y_{2})=0$ (cf. (\ref{eq:E[G]=0})), and $%
y_{2}-y_{1}=1/N.$
\end{itemize}

The above construction amounts to linearly interpolating the function $%
e_{t-1}$ from the finite grid $D$ to the whole interval $[0,1],$ and then
taking a point where this function vanishes ($y$ in Case 1, and $%
p_{1}y_{1}+p_{2}y_{2}$ in Case 2) and using it for the forecast ($y$ itself
in Case 1, and the $p_{1},p_{2}$ probabilistic mixture of $y_{1}$ and $y_{2}$
in Case 2). We thus have $\mathbb{E}_{t-1}\left[ e_{t-1}(c_{t})\right] =0$
in both cases, where $\mathbb{E}_{t-1}$ stands for $\mathbb{E}\left[ \cdot
|h_{t-1}\right] .$

\begin{theorem}
\label{th:1-dim-e}The above procedure $\sigma $ is $1/(2N)$-almost
deterministic and $1/(2N)$-calibrated.
\end{theorem}

\begin{proof}
Put $\bar{y}:=y$ in Case 1 and $\bar{y}:=(y_{1}+y_{2})/2$ in Case 2. Then $|%
\bar{y}-c_{t}|\leq 1/(2N)$ in both cases, which implies that $\mathbb{E}%
_{t-1}\left[ e_{t-1}(c_{t})\cdot (\bar{y}-c_{t})\right] \leq (1/2N)\mathbb{E}%
_{t-1}\left[ |e_{t-1}(c_{t})|\right] .$ Now $\mathbb{E}_{t-1}\left[
e_{t-1}(c_{t})\cdot (a-\bar{y})\right] =0$ for every $a$ (because $a-\bar{y}$
is constant given $h_{t-1},$ and $\mathbb{E}_{t-1}\left[ e_{t-1}(c_{t})%
\right] =0$ by the construction of $\sigma )$; adding to the previous
inequality gives the $(\Pi ,\varepsilon )$-stochastic forecast-hedging
condition (\ref{eq:SFH}), where $\Pi $ is the same as in the proof of
Theorem \ref{th:all-calib}(S), and $\varepsilon =1/(2N).$ Therefore $\sigma $
is $1/(2N)$-calibrated by Theorem \ref{th:FH}(S); in addition, $\sigma $ is $%
1/(2N)$-almost deterministic because we always have $|c_{t}-\bar{y}|\leq
1/(2N).$
\end{proof}

\bigskip

The calibration bound of $1/(2N)$ is the best that one can achieve with
forecasts on the grid $D$: consider for instance the action sequence where $%
a_{t}$ equals $1$ with probability $1/(2N),$ independently over $t.$

\section{Calibration and Game Dynamics\label{s:dynamics}}

Forecasts are a useful tool for dynamic multi-player interactions. Consider
a game that is played repeatedly. A natural type of game dynamic is one
where in each period the players make forecasts on what will happen next and
then choose their actions in response to these forecasts. Interesting
long-run behavior obtains when the forecasts are \textquotedblleft
good"---i.e., calibrated---and the responses to the forecasts are
\textquotedblleft good"---i.e., best responses.

The \textquotedblleft calibrated learning" of Foster and Vohra (1997), on
the one hand, and the \textquotedblleft publicly calibrated learning" of
Kakade and Foster (2004) and the \textquotedblleft smooth calibrated
learning" of Foster and Hart (2018), on the other hand, are two such types
of game dynamics. The main difference between the two types is that in the
former each player uses a stochastic classically calibrated forecasting
procedure, whereas in the latter all players use the same deterministic
weakly, or smoothly, calibrated forecasting procedure. In the long run, the
former yields correlated equilibria as the time average of play, whereas the
latter yields Nash equilibria as the period-by-period behavior (of course,
everything should be understood with appropriate \textquotedblleft
approximate" adjectives); see Foster and Hart (2018) for a more extensive
discussion. If we replace the deterministic weakly and smoothly calibrated
procedures with the stronger, but easier to obtain, deterministic
continuously calibrated procedures (see Proposition \ref{p:cc=>smooth,weak}
in Appendix \ref{sus-a:cont-calib}), we obtain the same long-run result:
period-by-period behavior that is close to Nash equilibria. The simplicity
of continuous calibration allows for a simple result and proof; see Theorem %
\ref{th:cont learn} below.

The game dynamics results underscore the importance of \emph{deterministic}
procedures, which are \textquotedblleft leaky" (see Foster and Hart 2018)
and thus remain calibrated even if in each period the forecast is revealed
before the action is chosen. By contrast, stochastic procedures are no
longer calibrated if the actual \emph{realization} of the random forecast is
revealed before the action is chosen.

\subsection{Continuously Calibrated Learning\label{sus:cont-learn}}

A finite \emph{game} is given by a finite set of players $N$, and, for each
player $i\in N,$ a finite set of pure strategies $A^{i}$ and a payoff
function $u^{i}:A\rightarrow \mathbb{R},$ where $A:=\prod_{i\in N}A^{i}$
denotes the set of strategy combinations of all players. Let $n:=|N|$ be the
number of players, $m^{i}:=|A^{i}|$ the number of pure strategies of player $%
i,$ and $m:=\sum_{i\in N}m^{i}.$ The set of mixed strategies of player $i$
is $X^{i}:=\Delta (A^{i}),$ the unit simplex (i.e., the set of probability
distributions) on $A^{i}$; we identify the pure strategies in $A^{i}$ with
the unit vectors of $X^{i},$ and so $A^{i}\subseteq X^{i}.$ Put $C\equiv
X:=\prod_{i\in N}X^{i}$ for the set of mixed-strategy combinations (i.e., $N$%
-tuples of mixed strategies). The payoff functions $u^{i}$ are multilinearly
extended to $X,$ and thus $u^{i}:X\rightarrow \mathbb{R}.$

For each player $i$ and combination of mixed strategies of the other players 
$x^{-i}=(x^{j})_{j\neq i}\in \prod_{j\neq i}X^{j}=:X^{-i},$ let $\bar{u}%
^{i}(x^{-i}):=\max_{y^{i}\in X^{i}}u^{i}(y^{i},x^{-i})=\max_{a^{i}\in
A^{i}}u^{i}(a^{i},x^{-i})$ be the maximal payoff that $i$ can obtain against 
$x^{-i};$ for every $\varepsilon \geq 0,$ let $\mathrm{BR}_{\varepsilon
}^{i}(x^{-i}):=\{x^{i}\in X^{i}:u^{i}(x^{i},x^{-i})\geq \bar{u}%
^{i}(x^{-i})-\varepsilon \}$ denote the set of $\varepsilon $\emph{-best
replies} of $i$ to $x^{-i}.$ A (mixed) strategy combination $x\in X$ is a 
\emph{Nash }$\varepsilon $\emph{-equilibrium} if $x^{i}\in \mathrm{BR}%
_{\varepsilon }^{i}(x^{-i})$ for every $i\in N;$ let \textrm{NE}$%
(\varepsilon )\subseteq X$ denote the set of Nash $\varepsilon $-equilibria
of the game.

A (discrete-time) \emph{dynamic }consists of each player $i\in N$ playing a
pure strategy $a_{t}^{i}\in A^{i}$ at each time period $t=1,2,...;$ put $%
a_{t}=(a_{t}^{i})_{i\in N}\in A.$ There is perfect monitoring: at the end of
period $t$ all players observe $a_{t}$. The dynamic is \emph{uncoupled }%
(Hart and Mas-Colell 2003, 2006, 2013) if the play of every player $i$ may
depend \emph{only} on player $i$'s payoff function $u^{i}$ (and not on the
other players' payoff functions). Formally, such a dynamic is given by a
mapping for each player $i$ from the history $h_{t-1}=(a_{1},...,a_{t-1})$
and his own payoff function $u^{i}$ into $X^{i}=\Delta (A^{i})$ (player $i$%
's choice may be random); we will call such mappings \emph{uncoupled.} Let $%
x_{t}^{i}\in X^{i}$ denote the mixed action that player $i$ plays at time $%
t, $ and put $x_{t}=(x_{t}^{i})_{i\in N}\in X.$

The dynamics we consider are continuous variants of the \textquotedblleft
calibrated learning" introduced by Foster and Vohra (1997). \emph{Calibrated
learning} consists of each player best replying to calibrated forecasts on
the other players' strategies; it results in the joint distribution of play
(i.e., the time average of the $N$-tuples of strategies $a_{t}$) converging
in the long run to the set of correlated equilibria of the game. We consider 
\emph{continuously calibrated learning}, where stochastic classic
calibration is replaced with deterministic continuous calibration, and best
replying is replaced with continuous approximate best replying. Moreover,
the forecasts are now $N$-tuples of mixed strategies (in $\prod_{i}\Delta
(A^{i})$), rather than correlated mixtures (in $\Delta (\prod_{i}A^{i})$).

Formally, given $\varepsilon >0$ a \emph{continuously calibrated }$%
\varepsilon $-\emph{learning} dynamic is given by:

\begin{enumerate}
\item[(I)] A deterministic continuously calibrated procedure on $X$, which
yields at each time $t$ a forecast $c_{t}=(c_{t}^{i})_{i\in N}\in X$ on the
distribution of strategies of each player$.$

\item[(II)] For each player $i\in N$ a continuous $\varepsilon $-best-reply
function $\beta ^{i}:X\rightarrow X^{i};$ i.e., $\beta ^{i}(x)\in \mathrm{BR}%
_{\varepsilon }^{i}(x^{-i})$ for every $x^{-i}\in X^{-i}.$
\end{enumerate}

The dynamic consists of each player running the procedure in (I), generating
at time $t$ a forecast $c_{t}\in X;$ then each player $i$ plays at period $t$
the mixed strategy\footnote{%
Thus $\mathbb{P}\left[ a_{t}=a~|~h_{t-1}\right] =\prod_{i\in
N}x_{t}^{i}(a^{i})$ for every $a=(a^{i})_{i\in N}\in A,$ where $h_{t-1}$ is
the history and $x_{t}^{i}(a^{i})$ is the probability that $x_{t}^{i}\in
\Delta (A^{i})$ assigns to the pure strategy $a^{i}\in A^{i}.$} $%
x_{t}^{i}:=\beta ^{i}(c_{t})\in X^{i},$ where $\beta ^{i}$ is given by (II).
All players observe the strategy combination $a_{t}=(a_{t}^{i})_{i\in N}\in
A $ that has actually been played, and remember it. Let $\beta (x)=(\beta
^{i}(x))_{i\in N};$ thus, $\beta :X\rightarrow X$ is a continuous function.
We refer to $c_{t}\in X$ as the \emph{forecast}, $x_{t}=\beta (c_{t})\in X$
the \emph{behavior} (i.e., the mixed strategies played), and $a_{t}\in A$
the \emph{actions} (i.e., the realized pure strategies played ($c_{t},x_{t},$
and $a_{t}$ depend on the history).

Since for each player $i$ the approximate best reply condition in (II) makes
use \emph{only} of player $i$'s payoff function $u^{i},$ we can without loss
of generality choose $\beta ^{i}$ so as to depend only on $u^{i},$ which
makes the dynamic \emph{uncoupled} (see above).

The existence of a deterministic continuously calibrated procedure in 1 is
given by Theorem \ref{th:all-calib}(D)$;$ the existence of $\varepsilon $%
-approximate continuous best-reply mappings in (II) is well known$.$

Our result is:

\begin{theorem}
\label{th:cont learn}Let $\Gamma =(N,(A^{i})_{i\in N},(u^{i})_{i\in N})$ be
a finite game. For every $\varepsilon >0,$ a continuously calibrated $%
\varepsilon $-learning dynamic is an uncoupled dynamic and satisfies almost
surely 
\begin{equation}
\lim_{t\rightarrow \infty }\frac{1}{t}\left\vert \{s\leq t:x_{s}\in \mathrm{%
NE}(\varepsilon ^{\prime })\}\right\vert =1  \label{eq:NE(eps')}
\end{equation}%
for every\footnote{%
It does \emph{not} follow that we can take $\varepsilon ^{\prime
}=\varepsilon ;$ for instance, consider the case where at time $t$ we have
an $(\varepsilon +1/t)$-equilibrium. \textquotedblleft Almost surely"
applies to all $\varepsilon ^{\prime }>\varepsilon $ simultaneously (take a
sequence $\varepsilon _{n}^{\prime }$ decreasing to $\varepsilon ).$} $%
\varepsilon ^{\prime }>\varepsilon .$
\end{theorem}

\bigskip

The proof goes by the following three claims. (i) If the forecasts $c_{t}$
are continuously calibrated for the sequence of pure strategies $a_{t},$
they are continuously calibrated also for the sequence of mixed strategies $%
x_{t}$ (because, by the law of large numbers, the long-run averages of the $%
a_{t}$'s and of the $x_{t}$'s are close, as $x_{t}$ is the expectation of $%
a_{t}$ conditional on the history). (ii) For every $c$, in every period
where the forecast is $c$ the mixed play is the same, namely, $x=\beta (c),$
and so if the sequence $c_{t}$ is continuously calibrated for the sequence $%
x_{t}$ then $c_{t}\approx x_{t}=\beta (c_{t}).$ (iii) From $c_{t}\approx
x_{t}$ we immediately get $x_{t}=\beta (c_{t})\approx \beta (x_{t})$ (apply
the continuous map $\beta $ to both sides), which says that the approximate
best reply to $x_{t}$ is $x_{t}$ itself, and thus $x_{t}$ is an approximate
Nash equilibrium.

The crucial feature of our dynamic is that continuous calibration is
preserved despite the fact that the actions depend on the forecasts (this
leakiness property does not hold for classic, probabilistic, calibration);
in addition, in each period all players have the same (deterministic)
forecast.

In Appendix \ref{sus-a:cont-learn} we provide a number of comments and
extensions.

\bigskip

\begin{proof}
For every $w:X\rightarrow \lbrack 0,1]$ let $\tilde{g}_{t}(w)$ be the
per-period gap for the mixed $x_{t}$ instead of the pure $a_{t},$ i.e.,%
\begin{equation}
\tilde{g}_{t}(w):=\frac{1}{t}\sum_{s=1}^{t}w(c_{s})(x_{s}-c_{s})=g_{t}(w)+%
\frac{1}{t}\sum_{s=1}^{t}w(c_{s})(a_{s}-x_{s}).  \label{eq:g-tilde}
\end{equation}

$\bullet $ \emph{Claim (i).} Let $W_{0}$ be a countable collection of
continuous functions $w:X\rightarrow \lbrack 0,1].$ Then for almost all
infinite histories $h_{\infty }=(c_{t},a_{t})_{t=1}^{\infty }$ we have%
\footnote{%
One can show, as in Section \ref{sus:binning-cc}, that $\lim_{t}\tilde{g}%
_{t}(w)=0$ for all continuous $w:X\rightarrow \lbrack 0,1]$ holds for almost
all infinite histories (however, there is no uniformity over the action
sequences).}%
\begin{equation*}
\lim_{t\rightarrow \infty }\tilde{g}_{t}(w)=0\;\;\text{for all }w\in W_{0}.
\end{equation*}

\emph{Proof}. First, for every $h_{\infty }$ we have $\lim_{t\rightarrow
\infty }g_{t}(w)=0$\ for all $w\in W_{0}$ by continuous calibration (see
Proposition \ref{p:cc<=>g}).

Second, for each $w$ we have $\mathbb{E}\left[ w(c_{s})a_{s}~|~h_{s-1}\right]
=w(c_{s})\mathbb{E}\left[ a_{s}~|~h_{s-1}\right] =w(c_{s})x_{s}$ (given $%
h_{s-1}$ the forecast $c_{s},$ and thus $w(c_{s}),$ is determined, and so
only $a_{s}$ is random; its conditional expectation is\footnote{%
Recall that we identify the pure actions $a^{i}\in A^{i}$ with the unit
vectors in the simplex $X^{i}.$} $\mathbb{E}\left[ a_{s}~|~h_{s-1}\right]
=\beta (c_{s})=x_{s}$). The Strong Law of Large Numbers for Dependent Random
Variables (Theorem 32.1.E in Lo\`{e}ve, 1978) says that 
\begin{equation}
\lim_{t\rightarrow \infty }\frac{1}{t}\sum_{s=1}^{t}\left( Y_{s}-\mathbb{E}%
\left[ Y_{s}|h_{s-1}\right] \right) =0\;\;\text{(a.s.)}  \label{eq:slln}
\end{equation}%
for bounded random variables $Y_{t};$ since the $w(c_{s})a_{s}$ are all
bounded by $\gamma ,$ and there are countably many $w$ in $W_{0}$, we obtain%
\begin{equation*}
\lim_{t\rightarrow \infty }\frac{1}{t}%
\sum_{s=1}^{t}(w(c_{s})a_{s}-w(c_{s})x_{s})=0\text{\ \ for all }w\in W_{0}\;%
\text{(a.s.).}
\end{equation*}%
Using (\ref{eq:g-tilde}) yields the claim. $\square $

$\bullet $ \emph{Claim (ii).} For every $\delta >0$ we have%
\begin{equation*}
\lim_{t\rightarrow \infty }\frac{1}{t}\left\vert \{s\leq t:\left\Vert \beta
(c_{s})-c_{s}\right\Vert \geq \delta \}\right\vert =0
\end{equation*}%
for almost every $h_{\infty }.$

\emph{Proof.} For every $d\in X$ and $\ell >0$ let $w_{d,\ell }(x):=[1-\ell
\left\Vert x-d\right\Vert ]_{+}$ (a \textquotedblleft tent" function on $X);$
thus, $w_{d,\ell }(x)>0$ if and only if $x\in B(d;1/\ell ).$ Let $D$ be the
set of points in $X$ with rational coordinates; put $W_{0}:=\{w_{d,\ell
}:d\in D,\ell \geq 1\};$ then $W_{0}$ is a countable collection of
continuous functions from $X$ to $[0,1],$ and so Claim (i) applies to it.

Take $\delta >0;$ the function $\alpha (x):=\beta (x)-x$ is uniformly
continuous on the compact set $X,$ and so there is an integer $\ell >0$ such
that $\left\Vert x-y\right\Vert \leq 1/\ell $ implies $\left\Vert \alpha
(x)-\alpha (y)\right\Vert \leq \delta .$ If $d\in D$ satisfies $\left\Vert
\alpha (d)\right\Vert \geq 2\delta ,$ then for every $x$ with $w_{d,\ell
}(x)>0,$ i.e., $x\in B(d;1/\ell ),$ we have $\left\Vert \alpha (x)-\alpha
(d)\right\Vert \leq \delta ,$ which yields%
\begin{equation*}
\left\Vert \sum_{s=1}^{t}w_{d,\ell }(c_{s})\alpha
(c_{s})-\sum_{s=1}^{t}w_{d,\ell }(c_{s})\alpha (d)\right\Vert \leq \delta
\sum_{s=1}^{t}w_{d,\ell }(c_{s}),
\end{equation*}%
that is, $\left\Vert t\tilde{g}_{t}(w_{d,\ell })-\alpha (d)n_{t}(w_{d,\ell
})\right\Vert \leq \delta n_{t}(w_{d,\ell }).$ Therefore 
\begin{equation*}
\left\Vert t\tilde{g}_{t}(w_{d,\ell })\right\Vert \geq \left( \left\Vert
\alpha (d)\right\Vert -\delta \right) n_{t}(w_{d,\ell })\geq \delta
n_{t}(w_{d,\ell }).
\end{equation*}
By Claim (i), this implies that 
\begin{equation}
\frac{1}{t}n_{t}(w_{d,\ell })\rightarrow 0  \label{eq:n->0}
\end{equation}%
almost surely as $t\rightarrow \infty .$

Take a finite set $D_{0}\subset D$ such that $\cup _{d\in D_{0}}B(d;1/\ell
)\supset X,$ and put $D_{1}:=\{d\in D_{0}:\left\Vert \alpha (d)\right\Vert
\geq 2\delta \}.$ The compact set $Y:=\{x\in X:\left\Vert \alpha
(x)\right\Vert \geq 3\delta \}$ is covered by $\cup _{d\in D_{1}}B(\delta
;1/\ell )$ (because $\left\Vert \alpha (x)\right\Vert \geq 3\delta $ implies
that there is $d\in D_{0}$ such that $y\in B(d;1/\ell ),$ and then $%
\left\Vert \alpha (d)\right\Vert \geq \left\Vert \alpha (x)\right\Vert
-\delta \geq 2\delta ),$ and the continuous function $\sum_{d\in
D_{1}}w_{d,\ell }(x)$ is positive on $Y,$ and thus it is $\geq \eta $ for
some $\eta >0,$ yielding%
\begin{equation*}
\sum_{d\in D_{1}}n_{t}(w_{d,\ell })=\sum_{s=1}^{t}\sum_{d\in D_{1}}w_{d,\ell
}(c_{s})\geq \eta \cdot \left\vert \{s\leq t:\left\Vert \alpha
(c_{s})\right\Vert \geq 3\delta \}\right\vert .
\end{equation*}%
Using (\ref{eq:n->0}) and replacing $\delta $ with $\delta /3$ completes the
proof$.$ $\square $

$\bullet $ \emph{Claim (iii).} For every $\varepsilon ^{\prime }>\varepsilon 
$ there is $\delta >0$ such that $\left\Vert \beta (c)-c\right\Vert \leq
\delta $ implies that $\beta (c)$ is a Nash $\varepsilon ^{\prime }$%
-equilibrium.

\emph{Proof.} By the uniform continuity of the functions $\beta ^{i}$ and $%
u^{i},$ let $\delta >0$ be such that $\left\Vert x-y\right\Vert \leq \delta $
implies $\left\vert u^{i}(\beta ^{i}(x),x^{-i})-u^{i}(\beta
^{i}(y),x^{-i})\right\vert \leq \varepsilon ^{\prime }-\varepsilon $ for
every $i.$ Taking $x=\beta (c)$ and $y=c$ yields $|u^{i}(\beta
^{i}(x),x^{-i})-u^{i}(x)|\leq \varepsilon ^{\prime }-\varepsilon $, which
together with $u^{i}(\beta ^{i}(x),x^{-i})\geq
\max_{y^{i}}u^{i}(y^{i},x^{-i})-\varepsilon $ by the choice of $\beta ^{i}$
as an $\varepsilon $-best reply proves the claim. $\square $

The theorem follows from Claims (ii) and (iii).
\end{proof}

\section{The Minimax Universe vs. the Fixed Point Universe\label{sus:fp-mm
universes}}

The forecast-hedging integration of the various calibration approaches that
we have carried out has pointed to a clear distinction between two separate,
parallel, universes: the \textsc{minimax} universe and the \textsc{fixed
point} universe.\footnote{%
This applies to dimension $m\geq 2$ (there is no distinction for dimension $%
m=1,$ where both minimax and fixed point reduce to the intermediate value
theorem).} Table 1 summarizes the differences exhibited in the present paper.%
%TCIMACRO{%
%\TeXButton{\arraystretch{1.5}}{\renewcommand{\arraystretch}{1.5}}}%
%BeginExpansion
\renewcommand{\arraystretch}{1.5}%
%EndExpansion

%TCIMACRO{\TeXButton{B}{\begin{table}[h] \centering}}%
%BeginExpansion
\begin{table}[h] \centering%
%EndExpansion
\begin{tabular}{c||c|c|}
& \textsc{minimax} & \textsc{fixed point} \\ \hline\hline
\multicolumn{1}{r||}{forecast-hedging} & stochastic & deterministic \\ \hline
\multicolumn{1}{r||}{procedure type} & MM & FP \\ \hline
\multicolumn{1}{r||}{calibration} & classic & continuous \\ \hline
\multicolumn{1}{r||}{equilibrium} & correlated & Nash \\ \hline
\multicolumn{1}{r||}{dynamic result} & time average & period-by-period \\ 
\hline
\end{tabular}%
\caption{The minimax and the fixed point universes}\label{TableKey}%
%TCIMACRO{\TeXButton{E}{\end{table}}}%
%BeginExpansion
\end{table}%
%EndExpansion

\appendix{}

\section{Appendix\label{s:appendix}}

\subsection{General Binnings\label{sus-a:cont-calib}}

In this appendix we show that the limitation to countable binnings is
without loss of generality.

Sums over arbitrary sets are defined, as usual, as the supremum over all
finite sums, i.e., $\sum_{i\in I}z_{i}:=\sup \{\sum_{i\in J}z_{i}:J\subseteq
I,$ $|J|<\infty \}$ (for real $z_{i}).$

Define a \emph{general binning} as $\Pi =(w_{i})_{i\in I},$ where $I$ is an
arbitrary set of bins and $w_{i}:C\rightarrow \lbrack 0,1]$ for every $i\in
I,$ such that $\sum_{i\in I}w_{i}(c)=1$ for every $c\in C.$ The general
binning $\Pi $ is \emph{continuous} if all $w_{i}$ are continuous functions.
The $\Pi $-calibration score is $K_{t}^{\Pi }:=\sum_{i\in I}\left\Vert
g_{t}(w_{i})\right\Vert .$

For classic calibration, $K_{t}$ is the maximal score, i.e.,%
\begin{equation*}
K_{t}=\max_{\Pi }K_{t}^{\Pi },
\end{equation*}%
where $\Pi $ ranges over all general binnings. Indeed, Lemma \ref{l:norm-W}
holds for arbitrary collections $(w_{j})_{j\in J}$ (apply it to finite sets
and then take the supremum), and so $K_{t}^{\Pi }\leq K_{t}$ for every
general binning $\Pi .$

For continuous calibration, which is defined as $\Pi $-calibration for every 
\emph{countable} continuous binning $\Pi ,$ we show that it implies $\Pi $%
-calibration for every continuous \emph{general} binning $\Pi $ as well.

\begin{proposition}
If the deterministic procedure $\sigma $ is continuously calibrated then it
is $\Pi $-calibrated for every continuous general binning $\Pi .$
\end{proposition}

\begin{proof}
Let $\Pi =(w_{i})_{i\in I}$ be a continuous general binning.

We claim that for every $\varepsilon >0$ there is a finite set $J^{\ast
}\subseteq I$ such that%
\begin{equation}
\left\Vert \sum_{i\in I\backslash J^{\ast }}w_{i}\right\Vert \leq
\varepsilon .  \label{eq:J*}
\end{equation}%
This follows from Dini's theorem for nets (instead of sequences); the proof
is the same, and as it is short we provide it here for completeness. Let $%
\mathcal{J}$ denote the collection of finite subsets of $I.$ For every $J\in 
\mathcal{J}$ let $D_{J}:=\{c\in C:\sum_{i\in J}w_{i}(c)>1-\varepsilon \};$
then $D_{J}$ is an open set (because $J$ is finite and so $\sum_{i\in
J}w_{i} $ is continuous), and $\cup _{J\in \mathcal{J}}D_{J}=C$ (because for
every $c $ we have $\sup_{J\in \mathcal{J}}\sum_{i\in J}w_{i}(c)=1,$ and so
there is $J\in \mathcal{J}$ for which the sum is $>1-\varepsilon ).$ The set 
$C$ is compact, and so there is a finite subcover $\cup
_{k=1}^{r}D_{J_{k}}=C.$ Put $J^{\ast }:=\cup _{k=1}^{r}J_{k};$ then $J^{\ast
}$ is a finite set, and $D_{J^{\ast }}=C$ (because $D_{J^{\ast }}\supseteq
D_{J_{k}}$ follows from $J^{\ast }\supseteq J_{k}).$ Thus for every $c\in C$
we have $\sum_{i\in J^{\ast }}w_{i}(c)>1-\varepsilon ,$ and so $\sum_{i\in
I\backslash J^{\ast }}w_{i}<\varepsilon ,$ which yields (\ref{eq:J*}).

Therefore, by Lemma \ref{l:norm-W}, 
\begin{equation*}
\sum_{i\in I\backslash J^{\ast }}\left\Vert g_{t}(w_{i})\right\Vert \leq
\varepsilon K_{t}\leq \gamma \varepsilon .
\end{equation*}%
For any $J\in \mathcal{J}$ we then have 
\begin{eqnarray*}
\sup_{\mathbf{a}_{t}}\sum_{i\in J}\left\Vert g_{t}(w_{i})\right\Vert &\leq
&\sum_{i\in J\cap J^{\ast }}\sup_{\mathbf{a}_{t}}\left\Vert
g_{t}(w_{i})\right\Vert +\sup_{\mathbf{a}_{t}}\sum_{i\in J\backslash J^{\ast
}}\left\Vert g_{t}(w_{i})\right\Vert \\
&\leq &\sum_{i\in J^{\ast }}\sup_{\mathbf{a}_{t}}\left\Vert
g_{t}(w_{i})\right\Vert +\gamma \varepsilon .
\end{eqnarray*}%
Taking the supremum over $J\in \mathcal{J}$ yields%
\begin{equation*}
\sup_{\mathbf{a}_{t}}\sum_{i\in I}\left\Vert g_{t}(w_{i})\right\Vert \leq
\sum_{i\in J^{\ast }}\sup_{\mathbf{a}_{t}}\left\Vert g_{t}(w_{i})\right\Vert
+\gamma \varepsilon ;
\end{equation*}%
the right-hand side converges to $\gamma \varepsilon $ as $t\rightarrow
\infty $ by (\ref{eq:G/t-pointwise}) of Proposition \ref{p:cc<=>g} (as $%
J^{\ast }$ is finite). Since $\varepsilon >0$ is arbitrary, the limit of the
left-hand side is $0.$
\end{proof}

\subsection{Continuous Calibration Implies Smooth and Weak Calibration\label%
{sus-a:smooth-calib}}

This appendix recalls the definitions of the existing concepts of smooth and
weak calibration, and proves that they are both implied by the stronger
concept of continuous calibration (see Section \ref{s:calibration-def}).

Let $\varepsilon \geq 0$ and $L<\infty .$ For a collection $\Lambda
=(\Lambda _{x})_{x\in C}$ of $L$-Lipschitz functions\footnote{%
A function $f$ is $L$-Lipschitz if $|f(z)-f(z^{\prime })|\leq L\left\Vert
z-z^{\prime }\right\Vert $ for all $z,z^{\prime }$ in the domain of $f.$} $%
\Lambda _{x}:C\rightarrow \lbrack 0,1],$ let%
\begin{equation}
\tilde{K}_{t}^{\Lambda }%
%TCIMACRO{\TeXButton{:=}{{\;:=\;}}}%
%BeginExpansion
{\;:=\;}%
%EndExpansion
\frac{1}{t}\sum_{x\in C}n_{t}(x)\left\Vert e_{t}(\Lambda _{x})\right\Vert .
\label{eq:K-Lambda}
\end{equation}%
A deterministic procedure is $(\varepsilon ,L)$-\emph{smoothly calibrated}
(Foster and Hart 2018) if%
\begin{equation*}
\varlimsup_{t\rightarrow \infty }\left( \sup_{\mathbf{a}_{t},\Lambda }\tilde{%
K}_{t}^{\Lambda }\right) \leq \varepsilon ,
\end{equation*}%
where the supremum is over all action sequences $\mathbf{a}$ and all
collections of $L$-Lipschitz functions $\Lambda =(\Lambda _{x})_{x\in C}$ as
above; it is $(\varepsilon ,L)$-\emph{weakly calibrated} (Kakade and Foster
2004, Foster and Kakade 2006) if%
\begin{equation*}
\varlimsup_{t\rightarrow \infty }\left( \sup_{\mathbf{a}_{t},w}\left\Vert
g_{t}(w)\right\Vert \right) \leq \varepsilon ,
\end{equation*}%
where the supremum is over all action sequences $\mathbf{a}$ and all $L$%
-Lipschitz functions $w:C\rightarrow \lbrack 0,1].$

While formula (\ref{eq:K-Lambda}) for $\tilde{K}_{t}^{\Lambda }$ resembles
formula (\ref{eq:K-Pi}) for $K_{t}^{\Pi },$ there are two differences. The
first is that the weight of $\left\Vert e_{t}(\Lambda _{x})\right\Vert $ in $%
\tilde{K}_{t}^{\Lambda }$ is \emph{not }the total weight $n_{t}(\Lambda
_{x}) $ of $\Lambda _{x}$ (which is the denominator of $e_{t}(\Lambda
_{x})), $ but rather the number of times $n_{t}(x)$ that $x$ has been used
as a forecast up to time $t$ (the sum in $\tilde{K}_{t}^{\Lambda }$ is thus
the finite sum over $x\in \{c_{1},...,c_{t}\}).$ The second is that the
functions $\Lambda _{x}$ do not form a binning; i.e., they do not add up to $%
\mathbf{1}.$ The second difference does not really matter (it can be
addressed, for instance, by rescaling the $\Lambda _{x}$ functions, which
does not affect the $e_{t}(\Lambda _{x}),$ because $e_{t}(w)$ is homogeneous
of degree $0$ in $w). $ The first difference is more significant; it
necessitates the use of certain approximations, such as the small cubes in
Lemma 11 in Foster and Hart (2018) and the resulting Proposition 13 there.%
\footnote{%
The bound on $\tilde{K}_{t}^{\Lambda }$ that is obtained in the proof of
Proposition 13 in Foster and Hart (2018) plays the same role as Proposition %
\ref{p:cc<=>g} here.}

By contrast, continuous calibration uses the more appropriate weights $%
n_{t}(\Lambda _{x});$ this streamlines the analysis and simplifies the
proofs. Moreover, continuous calibration yields a \textquotedblleft
universal" smoothly and weakly calibrated procedure for \emph{all} parameter
values $(\varepsilon ,L)$ at once (recall footnote \ref{ftn:doubling}).

\begin{proposition}
\label{p:cc=>smooth,weak}A deterministic procedure $\sigma $ that is
continuously calibrated is $(0,L)$-smoothly calibrated and $(0,L)$-weakly
calibrated for every $0<L<\infty .$
\end{proposition}

\begin{proof}
The convergence to zero in (\ref{eq:G/t-pointwise}) is uniform over any
finite set of continuous $w$'s, and thus, by (\ref{eq:norm-g}), over any
compact set of $w$'s---in particular, the set of $L$-Lipschitz functions $%
w:C\rightarrow \lbrack 0,1],$ which is compact by the Arzel\`{a}--Ascoli
theorem. This is precisely $(0,L)$-weak calibration; by Proposition 13 in
Foster and Hart (2018), it implies $(0,L)$-smooth calibration.
\end{proof}

\subsection{Outgoing Results\label{sus-a:outgoing}}

We provide here a number of comments and extensions to the results of
Section \ref{s:FH tools}.

\bigskip

\noindent \textbf{Remarks on Theorem \ref{th:hairy FP}.}

\emph{(a) }Theorem \ref{th:hairy FP} was proved using Brouwer's fixed point
theorem; conversely, Brouwer's theorem can be proved using Theorem \ref%
{th:hairy FP}. Indeed, let $g:C\rightarrow C$ be a continuous function.
Theorem \ref{th:hairy FP} applied to $f(x)=g(x)-x$ yields $y\in C$ such
that, in particular, $f(y)\cdot (g(y)-y)\leq 0$ (because $g(y)\in C);$ this
is $f(y)\cdot f(y)\leq 0,$ and so $f(y)=0,$ i.e., $g(y)=y.$

\emph{(b)} Brouwer's fixed point theorem is widely used to prove results in
many areas. Most such proofs use ingenious constructions, which are needed
to make the values of the continuous function lie in its domain, i.e., have
the function map $C$ \emph{into }$C.$ By contrast, Theorem \ref{th:hairy FP}
puts no restriction on the range of the function (beyond it being in the
Euclidean space of the same dimension); one only needs to ensure that a
point $y$ that satisfies (\ref{eq:OFP}) has the desired properties.

To demonstrate how Theorem \ref{th:hairy FP} may yield simpler proofs,
consider the famous result on the existence of Nash equilibria in finite
games (Nash 1950). Let $(N,(S^{i})_{i\in N},(u^{i})_{i\in N})$ be a finite
game in strategic form. Let $C:=\Pi _{i\in N}\Delta (S^{i})\subset \mathbb{R}%
^{m}$ where $m:=\sum_{i\in N}|S^{i}|,$ and, for every $x=(x^{i})_{i\in N}\in
C,$ put $f^{i}(x):=(u^{i}(s^{i},x^{-i}))_{s^{i}\in S^{i}}$ (this is the
vector of $i$'s payoffs for all his pure strategies against $x^{-i}),$ and $%
f(x):=(f^{i}(x))_{i\in N}.$ The function $f:C\rightarrow \mathbb{R}^{m}$ is
a polynomial and thus continuous, and so Theorem \ref{th:hairy FP} gives $%
y\in C$ such that $f(y)\cdot (c-y)\leq 0$ for every $c\in C.$ Taking in
particular $c=(x^{i},y^{-i})$ for any $i\in N$ and $x^{i}\in \Delta (S^{i}),$
we get $0\geq f(y)\cdot (c-y)=f^{i}(y)\cdot
(x^{i}-y^{i})=u^{i}(x^{i},y^{-i})-u^{i}(y^{i},y^{-i}),$ which shows that $y$
is a Nash equilibrium. Moreover, when the game is symmetric, putting $%
C:=\Delta (S^{1})$ and $f(x):=(u^{1}(s,x,...,x))_{s\in S}$ for every $x\in C$
yields the existence of a symmetric Nash equilibrium. Compare this short
proof to the usual proofs that are based directly on Brouwer's fixed point
theorem, which are much more intricate.

\bigskip

\noindent \textbf{Remarks on Theorem \ref{th:hairy mM}.}

\emph{(a)} The factor $\delta $ on the right-hand side of (\ref{eq:expfy})
can be lowered to $\delta _{0}\equiv \delta _{0}(D)<\delta $ (see the proof
of Theorem \ref{th:hairy mM}) by a limit argument, which is however no
longer a finite minimax construct. Indeed, take a sequence $B_{n}$ of finite 
$\delta _{n}$-grids of $C$ with $\delta _{n}$ decreasing to $0;$ we then get
a sequence of probability distributions $\eta _{n}\in \Delta (D)$ such that 
\begin{equation}
\mathbb{E}_{y\sim \eta _{n}}\left[ f(y)\cdot (x-y)\right] \leq (\delta
_{0}+\delta _{n})\,\mathbb{\mathbb{E}}_{y\sim \eta _{n}}\left[ \left\Vert
f(y)\right\Vert \right]  \label{eq:delta_n}
\end{equation}%
for every $n\geq 1$ and every $x\in C.$ Since $D$ is a finite set the
sequence $\eta _{n}$ has a limit point $\eta \in \Delta (D),$ say $\eta
_{n^{\prime }}\rightarrow \eta $ for a subsequence $n^{\prime }\rightarrow
\infty ;$ for each $x\in C$ taking the limit of (\ref{eq:delta_n}) as $%
n^{\prime }\rightarrow \infty $ then yields\footnote{%
The subsequence $n^{\prime }$ is such that $\eta _{n^{\prime }}(y)$ is a
convergent subsequence, with limit $\eta (y),$ for each one of the finitely
many elements $y$ of $D;$ then $\mathbb{E}_{y\sim \eta _{n^{\prime }}}\left[
g(y)\right] =\sum_{y\in D}\eta _{n^{\prime }}(y)g(y)\rightarrow \sum_{y\in
D}\eta (y)g(y)=\mathbb{E}_{y\sim \eta }\left[ g(y)\right] $ as $n^{\prime
}\rightarrow \infty $ for every real function $g$ on $D.$}%
\begin{equation}
\mathbb{E}_{y\sim \eta }\left[ f(y)\cdot (x-y)\right] \leq \delta _{0}\,%
\mathbb{\mathbb{E}}_{y\sim \eta }\left[ \left\Vert f(y)\right\Vert \right] .
\label{eq:delta0}
\end{equation}

\emph{(b)} The bound in (\ref{eq:delta0}) is tight: $\delta _{0}$ cannot be
lowered. Indeed, take a point $x_{0}\in C$ for which $\mathrm{dist}%
(x_{0},D)=\delta _{0},$ and consider the function $f:D\rightarrow \mathbb{R}%
^{m}$ defined by $f(y)=(x_{0}-y)/\left\Vert x_{0}-y\right\Vert $ for every $%
y\in D$; we have $\left\Vert f(y)\right\Vert =1$ and $f(y)\cdot
(x_{0}-y)=\left\Vert x_{0}-y\right\Vert \geq \delta _{0}$ for every $y\in D.$
\bigskip

\noindent \textbf{Remarks on Corollary \ref{c:outgoing-MM}.}

\emph{(a) }In Corollary \ref{c:outgoing-MM} one can get $\eta \in \Delta (C)$
with support of size at most $m+2$ (rather than $m+3),$ because when using
Carath\'{e}odory's theorem the last coordinate of $F(y),$ namely, $%
\left\Vert f(y)\right\Vert ,$ is no longer needed as it is replaced by the
constant $\sup_{x\in C}\left\Vert f(x)\right\Vert .$

\emph{(b) }If $f$ is a \emph{continuous} function then the result of
Corollary \ref{c:outgoing-MM} holds also for\footnote{%
Of course, Theorem \ref{th:hairy FP} yields in this case a stronger result,
i.e., a point $y$ rather than a distribution $\eta .$ However, the result
for $\varepsilon =0$ is obtained here by a minimax, rather than a fixed
point, theorem.} $\varepsilon =0.$ Indeed, take a sequence $\varepsilon
_{n}\rightarrow 0^{+}.$ For each $n,$ Corollary \ref{c:outgoing-MM} yields a
distribution $\eta _{n}$ on $C$ such that $\mathbb{E}_{y\sim \eta _{n}}\left[
f(y)\cdot (c-y)\right] \leq \varepsilon _{n}$ for every $c\in C.$ All the
distributions $\eta _{n}$ can be taken to have support of size at most $m+2$
(see Remark (a) above), and so the sequence $\eta _{n}$ has a limit point%
\footnote{%
Take a subsequence $n^{\prime }$ where all the $m+2$ values and all the $m+2$
probabilities converge (thus we do \emph{not} need to appeal to Prokhorov's
theorem); denote by $\eta $ the limit distribution. Then $\mathbb{E}_{y\sim
\eta _{n^{\prime }}}\left[ g(y)\right] \rightarrow _{n^{\prime }}\mathbb{E}%
_{y\sim \eta }\left[ g(y)\right] $ for any \emph{continuous} function $g$
(because then $p_{n^{\prime }}\rightarrow p$ and $y_{n^{\prime }}\rightarrow
y$ implies $p_{n^{\prime }}g(y_{n^{\prime }})\rightarrow pg(y)).$} $\eta ,$
which is also a distribution on $C$ with support of size at most $m+2;$ then 
$\mathbb{E}_{_{y\sim \eta }}\left[ f(y)\cdot (c-y)\right] \leq 0$ for every $%
c\in C$ (because $\eta _{n^{\prime }}\rightarrow \eta $ implies $\mathbb{E}%
_{y\sim \eta _{n^{\prime }}}\left[ f(y)\cdot (c-y)\right] \rightarrow
_{n^{\prime }}\mathbb{E}_{y\sim \eta }\left[ f(y)\cdot (c-y)\right] ,$ since 
$f(y)\cdot (c-y)$ is a continuous function of $y).$

\emph{(c)} If $f$ is \emph{not} continuous the result of Corollary \ref%
{c:outgoing-MM} need not hold for $\varepsilon =0;$ take for example $%
C=[0,2],$ and $f(x)=1$ if $x<1$ and $f(x)=-1$ if $x\geq 1.$ Assume that $%
\eta \in \Delta (C)$ satisfies $\mathbb{E}_{y\sim \eta }\left[ f(y)\cdot
(c-y)\right] \leq 0$ for all $c\in C.$ Taking $c=1$ gives $\mathbb{E}_{y\sim
\eta }\left[ f(y)\cdot (1-y)\right] \leq 0;$ but $f(y)\cdot (1-y)\geq 0$ for
all $y\in \lbrack 0,2],$ with equality only for $y=1,$ and so $\eta $ must
put unit mass on $y=1;$ but then $\mathbb{E}_{y\sim \eta }\left[ f(y)\cdot
(c-y)\right] =$ $f(1)\cdot (c-1)=1-c,$ which is positive for $c<1.$

\emph{(d) }The minimax theorem follows from Corollary \ref{c:outgoing-MM}.
First, consider a \emph{symmetric} finite two-person zero-sum game, given by
an $m\times m$ payoff matrix $B$ that is skew-symmetric (i.e., $B^{\top
}=-B).$ Take $C$ to be the unit simplex in $\mathbb{R}^{m}$ (i.e., the set
of mixed strategies), and let $f:C\rightarrow \mathbb{R}^{m}$ be given by $%
f(x):=Bx.$ Corollary \ref{c:outgoing-MM} together with Remark (b) above
implies that there exists a distribution $\eta $ on $C$ (with finite
support) such that $\mathbb{E}_{_{y\sim \eta }}\left[ y^{\top }B^{\top }(c-y)%
\right] =\mathbb{E}_{_{y\sim \eta }}\left[ By\cdot (c-y)\right] \leq 0$ for
every $c\in C.$ Now $y^{\top }B^{\top }y=0$ for every $y\in C$ by symmetry
(i.e., $B^{\top }=-B),$ and so $\mathbb{E}_{_{y\sim \eta }}\left[ y^{\top
}B^{\top }c\right] \leq 0$ for every $c\in C.$ Thus $z:=\mathbb{E}_{_{y\sim
\eta }}[y]\in C$ satisfies $z^{\top }Bc=-zB^{\top }c\geq 0$ for every $c\in
C,$ and so $z$ is a minimax strategy that guarantees the value $0;$ by
symmetry, $z$ is also a maximin strategy that guarantees the value $0,$ and
we are done. Finally, for a general two-person zero-sum game, use a standard
symmetrization argument (e.g., Luce and Raiffa 1957, A6.8).

\bigskip

\noindent \textbf{Remark on Theorem \ref{th:dhfp}.}

\emph{\ }If $C$ is a convex polytope and the set $D$ consists of the
vertices of a simplicial subdivision of $C,$ then we can define $\widetilde{f%
}$ by linearly interpolating inside each simplex; this implies that we
moreover have $\mathbb{E}_{y\sim \eta }\left[ y\right] =z$ (however, to keep
satisfying this additional property may require $\eta $ to have support of
size $2m+1$ instead of $m+1).$

\subsection{Deterministic and Stochastic Forecast-Hedging\label{sus-a:FH}}

We explain here why the proofs of (D) and (S) of Theorem \ref{th:FH} are
somewhat different: we use $S_{t}$ and the derived $\varphi _{t-1}$ in (D),
and $X_{t}$ and the derived $\psi _{t-1}$ in (S).

One can check that the $S_{t}$ approach in the (S) setup gives $%
\varlimsup_{t}(1/t^{2})S_{t}=\varlimsup_{t}\sum_{i=1}^{I}\left\Vert
g_{t}(w_{i})\right\Vert ^{2}\leq \varepsilon ^{2}.$ What this yields is $%
\varlimsup_{t}\mathbb{E}\left[ K_{t}^{\Pi }\right] =\varlimsup_{t}\mathbb{E}%
\left[ \sum_{i=1}^{I}||g_{t}(w_{i})||\right] \leq \varepsilon \sqrt{I}$
(consider for instance the case where the $||g_{t}(w_{i})||^{2}$ are all
equal to $\varepsilon ^{2}/I),$ which however does not suffice. Indeed, for
classic calibration the binning comes from an $\varepsilon $-grid of $C$
(see the proof of Theorem \ref{th:all-calib}(S)), and so its size $I$ is of
the order of $1/\varepsilon ^{m},$ which makes the bound $\varepsilon \sqrt{I%
}$ not useful beyond dimension $m=1.$ The more delicate approach with $X_{t}$
gets rid of this annoying $\sqrt{I}$ factor. The issue does not arise in
(D), since there we have $\varepsilon =0,$ and so $\varepsilon \sqrt{I}=0$
for every finite binning, which extends to countable continuous binnings by (%
\ref{eq:tail-gt}).

Going in the other direction, while we could use the $X_{t}$ approach for
(D) as well (it will not affect the result), the $S_{t}$ approach is
preferable as it is shorter and simpler.

\subsection{Calibration with Probability One\label{sus-a:prob1}}

In this appendix we show how to strengthen the results on classic
calibration (Theorem \ref{th:all-calib}(S) and (AD) in Section \ref%
{s:calibrated procedures}) from convergence in expectation to convergence
almost surely (\textquotedblleft a.s.").

The definition of classic calibration in Section \ref{sus:classic calib-def}
requires that the calibration score $K_{t}$ be small \emph{in expectation}
(i.e., that $\mathbb{E}\left[ K_{t}\right] $ be less than $\varepsilon $ in
the limit). One may require in addition that $K_{t}$ be small \emph{almost
surely }(i.e., with probability one); that is, for every action sequence $%
\mathbf{a,}$%
\begin{equation}
\varlimsup_{t\rightarrow \infty }K_{t}\leq \varepsilon \;\;\text{(a.s.).}
\label{eq:as}
\end{equation}%
We now show that the procedures constructed in Section \ref{s:calibrated
procedures} do indeed satisfy this additional requirement.

In the proof of Theorem \ref{th:FH}(S), the sequence $Y_{t}$ is uniformly
bounded (by $2\gamma \cdot \gamma +\gamma ^{2}=3\gamma ^{2}$), and so we can
apply the Strong Law of Large Numbers for Dependent Random Variables (see (%
\ref{eq:slln})): 
\begin{equation*}
\frac{1}{t}\sum_{s=1}^{t}\left( Y_{s}-\mathbb{E}\left[ Y_{s}|h_{s-1}\right]
\right) \rightarrow _{t\rightarrow \infty }0\;\;\text{(a.s.).}
\end{equation*}%
Since $\mathbb{E}\left[ Y_{s}|h_{s-1}\right] =\mathbb{E}_{s-1}\left[ Y_{s}%
\right] \leq \varepsilon ^{2}$ by (\ref{eq:Y-V}), it follows that $%
\varlimsup_{t\rightarrow \infty }(1/t)\sum_{s=1}^{t}Y_{s}\leq \varepsilon
^{2}$ (a.s.). Together with $\lim_{t\rightarrow \infty
}(1/t)\sum_{s=1}^{t}Z_{s}=0$ by (\ref{eq:Z}), we get $\varlimsup_{t%
\rightarrow \infty }(1/t)X_{t}\leq \varepsilon ^{2}$ (a.s.), and thus $%
\varlimsup_{t\rightarrow \infty }K_{t}^{\Pi }\leq \varepsilon $ (a.s.)
(because $(K_{t}^{\Pi })^{2}\leq (1/t)X_{t}).$ Applying this to the binning $%
\Pi $ of Theorem \ref{th:exist}(S) yields (\ref{eq:as}), for stochastic
classic calibration (Theorem \ref{th:all-calib}(S)) as well as for almost
deterministic classic calibration (Theorem \ref{th:all-calib}(AD)).

\subsection{Continuously Calibrated Learning\label{sus-a:cont-learn}}

In this appendix we provide a number of comments and extensions on the
result on game dynamics of Section \ref{s:dynamics}.

\textbf{Remarks on Theorem \ref{th:cont learn}. }\emph{(a) }The forecasts
are also approximate Nash equilibria:\footnote{%
Which is not surprising, as $c_{t}$ and $x_{t}$ are close (see Claim (ii)).
Of course, what we care about are not the forecasts, but the behaviors; this
is why the result in Theorem \ref{th:cont learn} is stated for $x_{t}.$}%
\begin{equation*}
\lim_{t\rightarrow \infty }\frac{1}{t}\left\vert \{s\leq t:c_{s}\in \mathrm{%
NE}(\varepsilon ^{\prime })\}\right\vert =1\text{\hspace{0.2in}}\mathrm{%
(a.s.)}
\end{equation*}%
for every $\varepsilon ^{\prime }>\varepsilon .$ This follows by replacing
Claim (iii) with: 

\emph{Claim (iii'). }For every $\varepsilon ^{\prime }>\varepsilon $ there
is $\delta >0$ such that $\left\Vert \beta (c)-c\right\Vert \leq \delta $
implies that $c\in \mathrm{NE}(\varepsilon ^{\prime })$ (for the proof, take 
$\delta >0$ such that $\left\Vert x-y\right\Vert \leq \delta $ implies $%
|u^{i}(x^{i},y^{-i})-u^{i}(y)|\leq \varepsilon ^{\prime }-\varepsilon $ for
every $i).$

\emph{(b) }A statement that is equivalent to (\ref{eq:NE(eps')}) is%
\begin{equation}
\lim_{t\rightarrow \infty }\frac{1}{t}\sum_{s=1}^{t}\mathrm{dist}(x_{s},%
\mathrm{NE}(\varepsilon ))=0,  \label{eq:dist(NE)}
\end{equation}%
which is the way it appears in Kakade and Foster (2004) (and the same
applies to the statement in (a) above). Indeed, for every $\varepsilon
^{\prime }>\varepsilon $ let $\delta (\varepsilon ^{\prime }):=\inf_{x\notin 
\mathrm{NE}(\varepsilon ^{\prime })}\mathrm{dist}(x,\mathrm{NE}(\varepsilon
))$ and $\rho (\varepsilon ^{\prime }):=\sup_{x\in \mathrm{NE}(\varepsilon
^{\prime })}\mathrm{dist}(x,\mathrm{NE}(\varepsilon ));$ then it is
straightforward to see that $\delta (\varepsilon ^{\prime })>0$ and $%
\lim_{\varepsilon ^{\prime }\searrow \varepsilon }\rho (\varepsilon ^{\prime
})=0$ (use the compactness of $X$ and the continuity of the functions $%
u^{i}).$ Therefore $\delta (\varepsilon ^{\prime })\mathbf{1}_{x\notin 
\mathrm{NE}(\varepsilon ^{\prime })}\leq $ $\mathrm{dist}(x,\mathrm{NE}%
(\varepsilon ))\leq \rho (\varepsilon ^{\prime })+\sqrt{m}\mathbf{1}%
_{x\notin \mathrm{NE}(\varepsilon ^{\prime })}$ (because $\sup_{x,y\in
X}\left\Vert x-y\right\Vert \leq \sqrt{m}).$ Using the first inequality for
each $x_{s}$ shows that (\ref{eq:dist(NE)}) implies (\ref{eq:NE(eps')}), and
using the second inequality for each $x_{s}$ shows that (\ref{eq:NE(eps')})
implies (\ref{eq:dist(NE)}) (the limit is $\leq \rho (\varepsilon ^{\prime
}) $ for every $\varepsilon ^{\prime }>\varepsilon ,$ and thus $0,$ because $%
\lim_{\varepsilon ^{\prime }\searrow \varepsilon }\rho (\varepsilon ^{\prime
})=0).$

\emph{(c)} The forecasting procedure in (I) depends \emph{only} on the sizes
of the strategy spaces $(m^{i})_{i\in N}.$

\emph{(d)} The play in each period $t$ need \emph{not} be independent across
the players, so long as the marginals are $(\beta ^{i}(c_{t}))_{i\in N}$.

\end{document}